\documentclass{article}
\usepackage[utf8]{inputenc}
\usepackage{multirow}
\usepackage{fullpage}
\usepackage{url}
\usepackage[colorlinks,linkcolor=Red,citecolor=Blue]{hyperref}
\usepackage[shortlabels]{enumitem}
\usepackage{amsmath}
\usepackage{amssymb}
\usepackage{amsthm}
\usepackage[dvipsnames]{xcolor}
\usepackage{mathtools}
\usepackage[utf8]{inputenc}
\usepackage{amsfonts}
\usepackage{graphicx}
\usepackage{mathrsfs}
\usepackage{physics}
\usepackage{complexity}
\usepackage{soul}
\usepackage{float}
\usepackage{bm} 
\usepackage[ruled,linesnumbered]{algorithm2e}
\usepackage[T1]{fontenc}
\usepackage[capitalize]{cleveref}
\usepackage{parskip}
\usepackage{comment}

\DeclareMathOperator{\Enc}{\mathsf{Enc}}

\DeclareMathOperator{\Dec}{\mathsf{Dec}}
\DeclareMathOperator{\test}{\mathsf{Test}}

\DeclareMathOperator{\Ver}{\mathsf{Ver}}

\DeclareMathOperator{\Eval}{\mathsf{Eval}}
\providecommand{\qsiO}{\pcalgostyle{qsiO}}
\DeclareMathOperator{\Obf}{\mathsf{Obf}}

\DeclareMathOperator{\calA}{\mathcal{A}}
\DeclareMathOperator{\calB}{\mathcal{B}}
\DeclareMathOperator{\calC}{\mathcal{C}}
\DeclareMathOperator{\calD}{\mathcal{D}}

\DeclareMathOperator{\calF}{\mathcal{F}}

\DeclareMathOperator{\calR}{\mathcal{R}}

\DeclareMathOperator{\calZ}{\mathcal{Z}}

\DeclareMathOperator{\Puncture}{\mathsf{Puncture}}
\DeclareMathOperator{\Supp}{\text{Supp}}
\DeclareMathOperator{\GL}{\mathsf{GL}}

\renewcommand{\CP}{\mathsf{CP}}
\renewcommand{\E}{\mathop{{}\mathbb{E}}}

\newcommand{\F}{\mathbb{F}}
\newcommand{\from}{\leftarrow}

\newcommand{\Adv}{\mathsf{Adv}}
\newcommand{\Red}{\mathsf{Red}}
\newcommand{\PRG}{\mathsf{PRG}}
\newcommand{\RandExpt}{\texttt{Rand-Expt}}
\newcommand{\SearchExpt}{\texttt{Search-Expt}}
\newcommand{\UEExpt}{\texttt{UE-Expt}}
\newcommand{\cUEExpt}{\texttt{cUE-Expt}}
\newcommand{\CPExpt}{\texttt{CP-Expt}}
\newcommand{\CPExptPF}{\texttt{CP-Expt-PtFunc}}
\newcommand{\CPExptDecision}{\texttt{CP-Expt-Decision}}
\newcommand{\CPExptSearch}{\texttt{CP-Expt-Search}}
\newcommand{\CPExptPRF}{\texttt{CP-Expt-PRF}}
\newcommand{\Hybrid}{\texttt{Hybrid}}

\newtheorem{theorem}{Theorem}

\newtheorem{definition}{Definition}
\newtheorem{lemma}[theorem]{Lemma}
\newtheorem{corollary}[theorem]{Corollary}

\newtheorem{const}{Construction}
\newtheorem{remark}{Remark}

\Crefname{Claim}{Claim}{Claims}
\Crefname{Game}{Game}{Games}

\usepackage[
	lambda,
	operators,
	advantage,
	sets,
	adversary,
	landau,
	probability,
	notions,	
	logic,
	ff,
	mm,
	primitives,
	events,
	complexity,
	asymptotics,
	keys]{cryptocode}

\createpseudocodeblock{pcb}{center,boxed}{}{}{}

\title{How to Use Quantum Indistinguishability Obfuscation}
\author{Andrea Coladangelo\thanks{Paul G. Allen School of Computer Science and Engineering, University of Washington. Email: \texttt{coladan@cs.washington.edu}.} \and Sam Gunn\thanks{UC Berkeley. Email: \texttt{gunn@berkeley.edu}. Supported by a Google PhD Fellowship.}}

\begin{document}

\maketitle

\begin{abstract}
    Quantum copy protection, introduced by Aaronson \cite{Aar09}, enables giving out a quantum program-description that cannot be meaningfully duplicated. Despite over a decade of study, copy protection is only known to be possible for a very limited class of programs.

    As our first contribution, we show how to achieve ``best-possible'' copy protection for all programs. We do this by introducing \emph{quantum state indistinguishability obfuscation} ($\qsiO$), a notion of obfuscation for \emph{quantum} descriptions of classical programs. We show that applying $\qsiO$ to a program immediately achieves best-possible copy protection.

    Our second contribution is to show that, assuming injective one-way functions exist, $\qsiO$ is concrete copy protection for a large family of puncturable programs --- significantly expanding the class of copy-protectable programs.
    A key tool in our proof is a new variant of unclonable encryption (UE) that we call \emph{coupled unclonable encryption} (cUE). While constructing UE in the standard model remains an important open problem, we are able to build cUE from one-way functions.
    If we additionally assume the existence of UE, then we can further expand the class of puncturable programs for which $\qsiO$ is copy protection.

    Finally, we construct $\qsiO$ relative to an efficient quantum oracle.
\end{abstract}

\newpage

\thispagestyle{empty}
{\small\tableofcontents}
\newpage

\section{Introduction}
A copy-protected program is one that can be evaluated by a user on arbitrary inputs, but not duplicated into a second, functionally equivalent program.
Since copy protection is impossible to achieve with classical information alone, Aaronson \cite{Aar09} proposed leveraging quantum information as a way to achieve provable copy protection.

Despite significant research, constructions of copy protection remain elusive. Even \emph{defining} copy protection is often quite subtle, with the right definition depending on the class of programs being copy protected. On the positive side, we know that copy protection can be achieved in either black-box models or for special classes of programs like pseudorandom functions and point functions \cite{AMP20, AP21, ALLZZ21, CLLZ21, AKL+22}. On the negative side, it is immediate that learnable programs cannot be copy protected \cite{Aar09}, and it is also known that there exist unlearnable programs that cannot be copy protected \cite{AP21}. Outside of these extremes, the landscape of copy protection remains poorly understood. For instance, our current understanding does not address copy protection for complex non-cryptographic software, e.g.\ video games. In general, the input/output behavior of a video game has almost no formal guarantees, so it seems difficult to achieve provable copy protection. This leads us to ask,
\begin{flalign} \label{question-1}
    && \textit{When are non-cryptographic programs copy protectable?} &&
\end{flalign}
A useful answer to this question should include conditions that can be heuristically verified in order to determine whether a given program is plausibly copy protectable.

Of course, even if a program \emph{can} be copy protected, it is not in general clear \emph{how} to copy-protect it. We would additionally like to know,
\begin{flalign} \label{question-2}
    && \textit{Is there a principled strategy for copy-protecting programs in general?} \tag{2} &&
\end{flalign}
In this work we introduce \emph{quantum state indistinguishability obfuscation} ($\qsiO$), which allows us to make progress on both of these questions.
To address Question (\ref{question-2}), we show that $\qsiO$ is \emph{optimal} copy protection for every class of programs. Therefore, assuming $\qsiO$ exists, Question (\ref{question-1}) reduces to determining which programs are actually copy protected by $\qsiO$.
We provide a partial answer to this question by showing that, roughly, copying a $\qsiO$ obfuscation is at least as hard as ``filling in'' the program on an input that has been redacted from the program description. 

\paragraph{Quantum state indistinguishability obfuscation ($\qsiO$).}
An obfuscator is an algorithm that takes as input a circuit $C$ and outputs an ``unintelligible'' program $C'$ with the same functionality as $C$ \cite{BGI+01}.

The most immediate generalization of this to the quantum setting is an obfuscator that takes as input a (classical description of) a quantum circuit $Q$ and outputs a (classical description of) a functionally equivalent quantum circuit $Q'$.

However, in this work we will be interested in encoding functionalities (classical or quantum) in \emph{quantum states}. In more detail, if $Q$ is a quantum circuit and $\rho$ is a quantum state, then we say that $(Q, \rho)$ is a \emph{quantum implementation} of a function $f$ if $\Pr[Q(\rho,x)=f(x)]=1$ for all $x$ in the domain of $f$. 

Several prior works have studied the question of whether obfuscators that are allowed to output quantum implementations are more powerful than obfuscators that can only output classical information, i.e.\ whether they can obfuscate a larger class of functionalities \cite{AF16,BK21,AP21,ABDS21, BM22, BKNY23}. However, all of these works consider obfuscators with classical input (and only the output is possibly a quantum state).

In contrast, a quantum state indistinguishability obfuscator $\Obf$ takes as input a quantum implementation of some function $f$, and outputs another quantum implementation of $f$. We say that $\Obf$ is a \emph{quantum state indistinguishability obfuscator} if, for any pair of quantum implementations $(Q_1, \rho_1)$ and $(Q_2, \rho_2)$ of the same function $f$, 
\[
    \Obf(Q_1, \rho_1) \approx \Obf(Q_2, \rho_2) 
\]
(where ``$\approx$'' denotes computational indistinguishability).
Note that we only consider obfuscation for $(Q, \rho)$ that implement some function $f$. In general, one could consider obfuscation for arbitrary quantum functionalities, but this is outside of the scope of our work.

\subsection{Our results}

\paragraph{Best-possible copy protection.}
The connection between $\qsiO$ and copy protection becomes clear through the observation that $\qsiO$ is \emph{best-possible} copy protection in the following (informal) sense: if a program $f$ can be copy protected, then obfuscating it using $\qsiO$ will copy-protect it. This follows from the fact that the $\qsiO$ obfuscation of a program is indistinguishable from the $\qsiO$ obfuscation of any copy-protected version of the program. Therefore, assuming $\qsiO$ exists, Question (\ref{question-1}) reduces to determining which $\qsiO$ obfuscations result in copy protection.

This result also directly addresses Question \eqref{question-2} by providing a universal heuristic to achieve copy protection. Furthermore, when using $\qsiO$ one does not need to worry about the subtleties that arise when defining copy protection for a particular class of programs; we are guaranteed that $\qsiO$ will achieve the best possible \emph{kind} of copy protection as well.

\paragraph{A construction of $\qsiO$ relative to a quantum oracle.}
In order to support the plausibility of $\qsiO$, we describe a proof-of-principle construction relative to an efficient quantum oracle. It is unclear how this quantum oracle can be heuristically instantiated --- however, it is often the case that such oracle constructions are the precursors to simpler instantiable constructions, or standard model constructions.

\paragraph{Copy protection for puncturable programs.}
The fact that $\qsiO$ is best-possible copy protection suggests that we should try to prove that it \emph{is} copy protection for certain classes of functions. We find that exploring conditions under which $\qsiO$ is copy protection sheds new light on Question \eqref{question-1} as well.

Assuming injective one-way functions, we show that $\qsiO$ copy-protects:
\begin{itemize}
    \item[(A)] Any puncturable program with ``indistinguishability'' at the punctured point.
    \item[(B)] Any puncturable program with ``non-reproducibility'' at the punctured point, under the additional assumption that unclonable encryption exists.
\end{itemize}
The idea of puncturing, along with techniques for how to use it, comes from \cite{SW21} where it is used extensively to build applications of classical $iO$. For convenience, we refer to puncturing with indistinguishability and non-reproducibility at the punctured point as \emph{decision} and \emph{search} puncturing, respectively. A puncturing procedure for a class of programs $\calF$ is an efficient algorithm $\Puncture$ that takes as input a description of a program $f \in \calF$ and a point $x \in \textnormal{Domain}(f)$, and outputs the description of a new program $f_x$. This program should satisfy $f_x(z) = f(z)$ for all $z \in \textnormal{Domain}(f) \setminus \{x\}$ as well as an additional security property:
\begin{itemize}
    \item For \emph{decision puncturing}, we require $(f_x, f(x)) \approx (f_x, f(x'))$ for a random $x'$. In \cite{SW21} it was shown that one-way functions imply the existence of decision puncturable pseudorandom functions.
    \item For \emph{search puncturing}, we require that no efficient adversary can compute from $f_x$ any output $y$ such that $\Ver(f,x,y) = 1$, for some efficient (public or private) verification procedure $\Ver$. For example, if $f$ is a signing function with a hard-coded secret key or a message authentication code, $\Ver(f,x,y)$ would use the verification key to check that $y$ is a valid signature or authentication tag for $x$. In \cite{BSW16} it was shown how to build search puncturable signing functions from indistinguishability obfuscation and one-way functions.
\end{itemize}

These results highlight some generic properties of programs that imply copy protectability, making progress on Question \eqref{question-1}: if a program can be described on \emph{all but one} input (i.e.\ it can be punctured), then in order to copy a $\qsiO$ obfuscation of the original program one must spend a comparable amount of work to that required to fill in the program's value at the missing point.

\paragraph{Techniques for the use of $\qsiO$.} One of the main contributions of this work is a technical toolkit for the use of $\qsiO$. The reader familiar with classical indistinguishability obfuscation ($\iO$) will recall that it is often used in conjunction with puncturing to obtain interesting applications. For $\qsiO$, we identify \emph{unclonable encryption} as the key primitive that, alongside puncturing, unlocks applications to copy protection. For $\qsiO$, we identify \emph{unclonable encryption}~\cite{BL20} as the key primitive that, alongside puncturing, unlocks applications to copy protection. Informally, unclonable encryption is a secret-key encryption scheme where ciphertexts are ``unclonable''.

As a key technical tool in our proof of (A), we introduce a new variant of unclonable encryption which we call \emph{coupled unclonable encryption}. Whereas constructing (full-fledged) unclonable encryption in the standard model remains an important open problem, we are able to build our variant from one-way functions,\footnote{If one is satisfied with encrypting messages of a \emph{fixed} polynomial length, then cUE exists unconditionally. This is a simple corollary of our result. However, in our applications of cUE, the messages are potentially much longer than the secret keys, and we therefore require a pseudorandom generator.} and we show that it suffices for (A). Given the notorious difficulty of building unclonable encryption in the standard model, we believe that our variant is of independent interest.

To further showcase our techniques, we show that assuming injective one-way functions and unclonable encryption, $\qsiO$ achieves a strong notion of copy protection for point functions which is beyond the reach of existing techniques.

\subsection{Comparison to previous work}

Two works are particularly related to ours: \cite{ALLZZ21}, which also studies copy protection for general programs; and \cite{CLLZ21}, which considers provable copy protection for specific functionalities that are similar to some of the ones we consider here.

\cite{ALLZZ21} takes a very different approach than ours to copy protection for general programs. By moving to a black-box model, they are able to build copy protection for \emph{all} unlearnable programs. However, it is known that there exist unlearnable programs that cannot be copy protected \cite{AP21}, so the black-box construction of \cite{ALLZZ21} does not address Question \eqref{question-1} about \emph{which} programs could be copy protectable. In contrast, $\qsiO$ could plausibly exist in the standard model for \emph{all programs}. Furthermore, we are able to identify specific properties that differentiate programs for which $\qsiO$ is copy protection.

While the black-box construction of \cite{ALLZZ21} does naturally suggest a heuristic copy protection scheme for arbitrary programs (by replacing black-box obfuscation with $\iO$), there is no ``best-possible'' guarantee comparable to $\qsiO$. There may exist programs that can be copy protected, and yet this heuristic construction nonetheless fails to copy-protect them. In order to address Question \eqref{question-1}, \cite{ALLZZ21} give a non-black-box construction of copy \emph{detection} for any watermarkable program, assuming public-key quantum money. They interpret this construction as evidence that copy \emph{protection} might exist for watermarkable programs as well.

\cite{CLLZ21} does not directly consider the problem of copy protection for general functionalities. Instead, one of the main results (under an information-theoretic conjecture that was later proven to be true in \cite{CV22}) is that punturable pseudorandom functions can be copy protected using $\iO$, assuming sub-exponentially-secure LWE. Compared to our provable copy protection results, the advantage of \cite{CLLZ21} is that $\iO$ is much more well-studied than $\qsiO$.\footnote{Despite significant research though, a construction of post-quantum $\iO$ from well-founded assumptions is still not known.} However, their result is limited to puncturable pseudorandom functions (and does not seem to extend further), while our results are applicable to a much broader class of puncturable functionalities. Additionally, our results do not rely on ``structured'' assumptions like LWE.


\subsection{Technical overview}
\label{subsec:techo}

\paragraph{Definitions.}
Throughout this technical overview, we will fix a universal quantum evaluation circuit $\Eval$. Instead of considering implementations as circuit-state pairs $(C, \rho)$, we will assume that the description of $C$ is included in $\rho$. Therefore we will view $\qsiO$ schemes as acting only on the quantum part, $\rho$.

As in the introduction, we say that $\rho$ \emph{implements} a function $f$ if, for all $x$, $\Pr[\Eval(\rho,x) = f(x)] = 1$ (or is negligibly close to $1$). An obfuscator $\Obf$ is a $\qsiO$ scheme if it satisfies:
\begin{itemize}
    \item (Correctness) if $\rho$ implements $f$, then $\Obf(\rho)$ implements $f$, and
    \item (Security) if $\rho, \rho'$ both implement $f$, then $\Obf(\rho) \approx \Obf(\rho')$.
\end{itemize}
We will write $\qsiO(\rho)$ to refer to a $\qsiO$ obfuscation of $\rho$.

\paragraph{Best-possible copy protection.}
With the definition of $\qsiO$ in hand, it is not difficult to prove that $\qsiO(f)$ is best-possible copy protection for any functionality $f$. Here is a sketch of the argument; for a more complete treatment see \Cref{theorem:best-possible}.

Let $\calF$ be any class of programs for which some copy protection scheme $\CP$ exists. That is, $\CP$ is an efficient quantum algorithm such that for $f \in \calF$, $\CP(f)$ outputs a quantum state $\rho$ such that $\Eval(\rho, x) = f(x)$ for all $x \in \textnormal{Domain}(f)$, and there is some guarantee of ``unclonability'' on $\rho$. It turns out that \Cref{theorem:best-possible} is not sensitive to the the precise definition of ``unclonability'' --- whatever definition of unclonability is satisfied by $\CP$, $\qsiO$ achieves the same guarantee. The key observation is that any adversary who wins the unclonability game for $\qsiO(f)$ must necessarily win the unclonability game for $\qsiO(\CP(f))$ as well, or else it would break the $\qsiO$ security guarantee! Since we can efficiently apply $\qsiO$ to $\CP(f)$ to prepare $\qsiO(\CP(f)) \approx \qsiO(f)$, it follows that $\qsiO(f)$ is at least as secure as $\CP(f)$.

\paragraph{Construction of $\qsiO$ relative to a quantum oracle}
Our construction of $\qsiO$ relative to a quantum oracle is simple, although the security proof is fairly involved. On input a quantum implementation $\rho$ of some function $f$, $\qsiO$ samples a uniformly random Clifford unitary $C$ and outputs the state $\tilde{\rho} = C \rho C^{\dagger}$, alongside an oracle implementing the unitary $G_C = C^{\dagger} \mathsf{Eval} C$, where $\mathsf{Eval}$ is a universal circuit. In other words, $\qsiO$ applies a Clifford one-time pad to the input state $\rho$; the oracle $G_C$ undoes the one-time pad, evaluates the function $f$, and then re-applies the one-time pad.

The ``Clifford twirl'' is sufficient to argue security against adversaries that make a \emph{single} query, but a more careful argument is required to handle general adversaries. This argument makes use of the ``admissible oracle lemma'' from \cite{GJMZ23}.

\paragraph{Unclonable encryption.}
As is often the case with classical $\iO$ \cite{SW21}, we find that $\qsiO$ does not by itself yield the applications we are most interested in. Instead, we combine $\qsiO$ with one-way functions and variants of unclonable encryption to build copy protection. We describe some background and a new result on unclonable encryption before discussing copy protection.

Unclonable encryption (UE), formally introduced by Broadbent and Lord \cite{BL20},\footnote{The notion was informally put forward by Gottesman in \cite{Got03}, who left constructing it as an open question. Broadbent and Lord \cite{BL20} formalized the notion, and achieved the first provably secure construction. We remark that Broadbent and Lord refer to what we call unclonable encryption as unclonable encryption with ``unclonable indistinguishability.''} can be viewed as an unclonable version of secret key encryption. A UE scheme consists of a generation algorithm that samples a classical secret key $\sk$, an encryption algorithm $\enc$ that outputs a quantum state, and a decryption algorithm $\dec$ that outputs a message. The security guarantee says that, without the secret key, an adversary given $\enc(\sk; m)$ cannot prepare two states which can later be used to decrypt the message $m$ (when provided the secret key $\sk$). We require UE schemes to have semantic security --- that is, the two states cannot both be used to learn non-negligible information about the message. Formally, a UE scheme $(\enc,\dec)$ is secure if no efficient adversary can win the following security game with probability noticeably greater than $1/2$:

\indent $\UEExpt(\secpar)$:
\begin{enumerate}
    \item The adversary sends the challenger a message $m$.
    \item The challenger samples a challenge bit $c \from \{0,1\}$ and a secret key $\sk \from \{0,1\}^\secpar$.
    \begin{enumerate}
        \item If $c=0$, the challenger samples a random message $r$ of the same length as $m$ and sends $\enc(\sk; r)$ to the adversary.
        \item If $c=1$, the challenger sends $\enc(\sk; m)$ to the adversary.
    \end{enumerate}
    \item The adversary splits into two non-communicating parties $A$ and $B$.
    \item The challenger sends each of $A$ and $B$ the secret key $\sk$.
    \item $A$ outputs a bit $a'$ and $B$ outputs a bit $b'$. The adversary wins if $a'=b'=c$.
\end{enumerate}

The first provably secure construction of UE was proposed in \cite{BL20}, and it satisfied a ``search-based'' notion of security in the quantum random oracle model (QROM). Subsequent work \cite{AKL+22, AKL23} achieved the ``decision'' version of UE that we consider here, still in the QROM. We conjecture that UE for single-bit messages can be built (for general messages) in the standard model, assuming one-way functions. 

One of the key insights of Broadbent and Lord \cite{BL20} is to link the ``search-based'' notion of UE to the following ``monogamy of entanglement'' result from \cite{TFKW13}, which says that 
no (unbounded) adversary can win the following security game with probability noticeably greater than 0:

\indent $\SearchExpt(\secpar)$:
\begin{enumerate}
    \item The challenger samples $x, \theta \from \{0,1\}^\secpar$ and sends $\ket{x^\theta}$ to the adversary. Here, $\ket*{x^\theta}$ is shorthand for $H^\theta \ket*{x}$, where $H^\theta$ denotes Hadamard gates applied to the qubits where the corresponding bit in $\theta$ is 1.
    \item The adversary splits into two non-communicating parties $A$ and $B$.
    \item The challenger sends each of $A$ and $B$ the basis $\theta$.
    \item $A$ and $B$ output strings $x_A, x_B$. The adversary wins if $x_A = x_B = x$.
\end{enumerate}

The reason that this result does not immediately yield UE (by using $x$ as a one-time pad for the message) is that the adversaries are required to guess \emph{all of} the message in $\SearchExpt$, whereas the adversaries in $\UEExpt$ are merely required to learn \emph{anything at all about} the message. For instance, if the adversary simply passes the first half of the qubits of $\ket{x^\theta}$ to $A$ and the second half to $B$, then both $A$ and $B$ can learn half of $x$. It is natural to attempt to evade this issue by using a randomness extractor. For a single-bit message $m$, we could use the following as a candidate unclonable encryption:
\begin{equation} \label{eq:candidate-single-bit-ue}
    \ket*{x^\theta},\, m \oplus u \cdot x
\end{equation}
where $x, \theta, u \from \{0,1\}^\secpar$, and the dot product $u \cdot x$ is taken over $\F_2$. The secret key is $\sk = (\theta, u)$, and the decryption algorithm simply reads $x$, computes $u \cdot x$, and removes the one-time pad on $m$.

Intuitively, it would seem that an adversary needs to learn all of $x$ in order to guess $u \cdot x$. This is typically proven using the quantum Goldreich-Levin reduction \cite{BV97,AC02}. Given a single quantum query to a predictor that successfully guesses $u \cdot x$ with probability $1/2+\varepsilon$ (over a random choice of $u$), the quantum Goldreich-Levin reduction produces a guess for the entire string $x$ with probability $\poly(\varepsilon)$. Since an adversary that wins $\UEExpt$ must have both parts $A$ and $B$ guess $u \cdot x$ correctly, we can run the quantum Goldreich-Levin reduction to show that each of $A$ and $B$ has at least a $\poly(\varepsilon)$ probability of guessing $x$. However, there is no guarantee that they guess $x$ correctly \emph{simultaneously}, so this reduction might never win $\SearchExpt$!

We do not know how to prove that the candidate UE scheme of \Cref{eq:candidate-single-bit-ue} is secure. Instead, we relax the requirement of UE so that a similar reduction works. This results in a variant of UE that we call \emph{coupled unclonable encryption} (cUE). In cUE, a ciphertext encrypts two messages under two independent secret keys. Each secret key alone works to decrypt the corresponding message. In the security game, $A$ receives one secret key, and $B$ receives the other. Our cUE encryption scheme for single-bit messages $m_A, m_B$ is:
\begin{equation} \label{eq:single-bit-cue}
    \ket*{x^\theta},\, m_A \oplus u \cdot x,\, m_B \oplus v \cdot x
\end{equation}
where $x, \theta, u, v \from \{0,1\}^\secpar$. The secret keys are $\sk_A = (\theta, u)$ and $\sk_B = (\theta, v)$. Now that $u$ and $v$ are independent, it is possible to prove that the above reduction works. Indeed, as we were working on this manuscript, similar ``simultaneous'' Goldreich-Levin theorems were proven in \cite{KT23, AKL23}. However, both of these works leave open the question of running a similar reduction for \emph{many-bit} messages. Specifically, in \cite{KT23}, the authors ask whether one can use many inner products to encrypt many bits, noting that their techniques do not extend to this setting. We answer this question in the affirmative in \Cref{subsec:unc-rand}, by carrying out a version of a ``hybrid argument'' on quantum operators.

This result is crucial for our copy protection applications, which require cUE for many-bit messages. Formally, the security guarantee of cUE states that an adversary cannot win the following game with probability noticeably greater than 1/2:

\indent $\cUEExpt(\secpar)$:
\begin{enumerate}
    \item The adversary sends the challenger two messages $m_A, m_B$.
    \item The challenger samples two challenge bits $a,b \from \{0,1\}$, two secret keys $\sk_A, \sk_B \from \{0,1\}^\secpar$, and two random messages $r_A, r_B$ of the same lengths as $m_A, m_B$, respectively.
    \item Let $m_A^0 = m_A, m_B^0 = m_B$, and $m_A^1 = r_A, m_B^1 = r_B$. The challenger sends $\enc(\sk_A, \sk_B; m_A^a, m_B^b)$ to the adversary.
    \item The adversary splits into two non-communicating parties $A$ and $B$.
    \item The challenger sends $\sk_A$ to $A$ and $\sk_B$ to $B$.
    \item $A$ outputs a bit $a'$ and $B$ outputs a bit $b'$. The adversary wins if $a'=a$ and $b'=b$.
\end{enumerate}
For general (many-bit) messages $m_A, m_B$, our cUE encryptions are essentially\footnote{This construction does not technically satisfy the syntax of $\cUEExpt$, because the secret keys $(\theta, U)$ and $(\theta, V)$ are not independent. This minor issue is resolved in \Cref{subsec:cue}.}
\begin{equation} \label{eq:many-bit-cue}
    \ket*{x^\theta},\, m_A \oplus \PRG(Ux),\, m_B \oplus \PRG(Vx).
\end{equation}
where $U, V$ are wide $\F_2$ matrices of appropriate dimensions, $Ux, Vx$ denote matrix-vector products, and $\PRG$ is any pseudorandom generator with appropriate stretch. Since the lengths of $Ux$ and $Vx$ are fixed as a function of $\secpar$, but the adversary can choose $m_A, m_B$ of whatever length it wishes, we need to use pseudorandom generators to potentially stretch $Ux$ and $Vx$ to the proper lengths.

We divide the proof of security for \Cref{eq:many-bit-cue} into two steps. First, in \Cref{subsec:unc-rand} we show that one of $Ux$ and $Vx$ is completely unpredictable to the corresponding pirate; we call this property \emph{unclonable randomness}. This is the core of the cUE proof and perhaps the most technical part of this work, requiring a new and delicate argument that resolves the aforementioned open question of \cite{KT23}. In \Cref{subsec:cue}, we invoke the security of the PRG to see that the cUE scheme is secure. Thus, assuming only the existence of one-way functions, there exists a cUE scheme that encrypts messages of arbitrary polynomial length.


In \Cref{subsec:decision-cp}, we show that cUE suffices to show that $\qsiO$ copy-protects puncturable programs with indistinguishability at the punctured point.

\begin{remark}
    In \cite{AKL+22}, the authors discuss ``issues with using extractors.'' The proposal for UE in \Cref{eq:candidate-single-bit-ue} falls within the category of extractor-based schemes that they are referring to, so the issues with natural proof techniques discussed there apply. However, the security of the UE scheme described above is not ruled out by their impossibility result (Theorem 1.3). Furthermore, our constructions of single-bit and general cUE in \Cref{eq:single-bit-cue,eq:many-bit-cue} are also extractor-based schemes in a similar sense, and we are nonetheless able to prove them secure. Therefore, we hope that our insights for constructing cUE may eventually be useful for constructing UE, as they may evade some of the barriers discussed in \cite{AKL+22}.
\end{remark}

Finally, we show that one can generically add a functionality that we call \emph{key testing} to any UE or cUE scheme, using $\qsiO$ and injective one-way functions. Key testing means that there is an algorithm $\test$ which determines whether a given string $z$ is a valid key for a given encryption $\sigma$. Key testing turns out to be crucial for our proofs of copy protection from $\qsiO$. The main idea to upgrade a UE or cUE scheme to one with key testing is to append to the ciphertext a $\qsiO$ obfuscation of the program $\delta_{\sk}$ (which is zero everywhere except at $\sk$). Intuitively, this allows one to check the validity of a secret key, while at the same time preserving unclonability thanks to the properties of $\qsiO$.

\paragraph{Copy protection for PRFs.}
Armed with cUE, we can apply $\qsiO$ to achieve copy protection for certain classes of functions. For the purposes of the technical overview, we will only describe how $\qsiO$ copy-protects pseudo-random functions (PRFs). This description highlights some of the main ideas behind our proof technique for the more general results of \Cref{sec:cp}. The basic idea of the proof technique is to use the $\qsiO$ guarantee to replace the PRF with a punctured version, where the values of the PRF at the challenge points are hard coded under a cUE encryption.

We explain this more precisely. Suppose that $\calF_\secpar$ is a family of puncturable PRFs with domain $\{0,1\}^{\secpar}$ and range $\{0,1\}^{n(\secpar)}$. It was shown in \cite{SW21} that puncturable PRFs can be built from any one-way function. We will prove that $\qsiO$ is a secure copy protection scheme for $\calF_\secpar$ via a sequence of hybrids, beginning with the PRF copy protection security game:

\indent $\CPExptPRF(\secpar)$:
\begin{enumerate}
    \item The challenger samples $f \from \calF_\secpar$, $a,b \from \{0,1\}$, $x_A, x_B \from \{0,1\}^{\secpar}$, and $y_A^0, y_B^0 \from \{0,1\}^{n(\secpar)}$. Let $y_A^1 = f(x_A)$ and $y_B^1 = f(x_B)$.
    \item The challenger sends the adversary $\qsiO(f)$.
    \item The adversary splits into two non-communicating parties $A$ and $B$.
    \item The challenger sends $x_A, y_A^a$ to $A$ and $x_B, y_B^b$ to $B$.
    \item $A$ outputs a bit $a'$ and $B$ outputs a bit $b'$. The adversary wins if $a'=a$ and $b'=b$.
\end{enumerate}
In other words, in this security game, the parties $A$ and $B$ are trying to decide whether they received a pair $(x,y)$ where $y = f(x)$ or where $y$ is uniformly random.

Let $f_{x_A,x_B}$ be $f$ punctured at $x_A, x_B$, let $\enc$ be a cUE scheme with key testing, and let
\[
    \sigma = \enc(x_A, x_B; f(x_A), f(x_B)).
\]
Our first hybrid uses the $\qsiO$ guarantee to replace $\qsiO(f)$ with $\qsiO(P[f_{x_A,x_B},\sigma])$, where $P[f_{x_A,x_B},\sigma]$ is a program (formally a \emph{quantum implementation} of a program) that does the following on input $z$:
\begin{enumerate}
    \item Use key testing to check whether $z$ is a valid key for $\sigma$. If not, terminate and output $f_{x_A,x_B}(z)$.
    \item Otherwise, use $z$ to decrypt $\sigma$ and output the result.
\end{enumerate}
Since $P[f_{x_A,x_B},\sigma](z) = f(z)$ for all $z$, $\qsiO(P[f_{x_A,x_B},\sigma]) \approx \qsiO(f)$. Therefore, the adversary's success probability in $\CPExptPRF(\secpar)$ does not change if the challenger instead sends $\qsiO(P[f_{x_A,x_B},\sigma])$ instead of $\qsiO(f)$ in step 2. Call this modified experiment $\Hybrid_1(\secpar)$.

Now, the pseudorandomness of $f$ at the punctured points implies that
\[
    (f_{x_A,x_B}, f(x_A), f(x_B), \enc(x_A, x_B; f(x_A), f(x_B))) \approx (f_{x_A,x_B}, \tilde{y}_A^1, \tilde{y}_B^1, \enc(x_A, x_B; \tilde{y}_A^1, \tilde{y}_B^1))
\]
where $\tilde{y}_A^1, \tilde{y}_B^1$ are random strings from the range of $f$. Therefore, the adversary's success probability is again preserved if we replace $f(x_A), f(x_B)$ with $\tilde{y}_A^1, \tilde{y}_B^1$ in $\Hybrid_1(\secpar)$. We also rename $y_A^0, y_B^0$ (introduced in step 1 of the original experiment) to $\tilde{y}_A^0, \tilde{y}_B^0$ for convenience  of notation. Then,  $\Hybrid_2(\secpar)$ is the following.

\indent $\Hybrid_2(\secpar)$:
\begin{enumerate}
    \item The challenger samples $f \from \calF_\secpar$, $a,b \from \{0,1\}$, $x_A, x_B \from \{0,1\}^{\secpar}$, and $\tilde{y}_A^0, \tilde{y}_B^0, \tilde{y}_A^1, \tilde{y}_B^1 \from \{0,1\}^{n(\secpar)}$.
    \item The challenger prepares $\tilde{\sigma} = \enc(x_A, x_B; \tilde{y}_A^1, \tilde{y}_B^1)$ and sends the adversary $\qsiO(P[f_{x_A,x_B}, \tilde{\sigma}])$.
    \item The adversary splits into two non-communicating parties $A$ and $B$.
    \item The challenger sends $x_A, \tilde{y}_A^a$ to $A$ and $x_B, \tilde{y}_B^b$ to $B$.
    \item $A$ outputs a bit $a'$ and $B$ outputs a bit $b'$. The adversary wins if $a'=a$ and $b'=b$.
\end{enumerate}
Our last hybrid, $\Hybrid_3(\secpar)$, will be the same as $\Hybrid_2(\secpar)$ except that the challenger sends the adversary $\qsiO(P[f,\tilde{\sigma}])$ instead of $\qsiO(P[f_{x_A,x_B},\tilde{\sigma}])$ in step 2. The adversary's success probability is negligibly close between $\Hybrid_2(\secpar)$ and $\Hybrid_3(\secpar)$ because $P[f,\tilde{\sigma}]$ and $P[f_{x_A,x_B},\tilde{\sigma}]$ are functionally equivalent, and so $\qsiO(P[f_{x_A,x_B},\tilde{\sigma}]) \approx \qsiO(P[f,\tilde{\sigma}])$.

Finally, notice that $\Hybrid_3(\secpar)$ is now quite close to the cUE experiment $\cUEExpt(\secpar)$! It's not difficult to see that there is a direct reduction from $\cUEExpt(\secpar)$ to $\Hybrid_3(\secpar)$, because $\qsiO(P[f, \tilde{\sigma}])$ can be generated from $\tilde{\sigma}$ by sampling $f \from \calF_\secpar$.

\subsection{Preliminaries}
We introduce some notation that we will use throughout the paper. 

We denote a \textit{quantum polynomial-time algorithm} with the acronym QPT. Formally, this is a polynomial-time uniform family of quantum circuits, where each circuit in the family is specified by a sequence of unitary operations and measurements. A quantum algorithm may in general receive (mixed) quantum states as inputs and produce (mixed) quantum states as outputs.

For a distribution $D$, the notation $x \from D$ denotes sampling an element from $D$; for a set $S$, $x \from S$ denotes sampling an element uniformly at random from $S$. For distributions $\calD, \calD'$, we write $\calD \approx \calD'$ and $\calD \equiv \calD'$ to indicate computational and statistical indistinguishability, respectively.

We denote by $\calC_d$ the Clifford group for dimension $d$, i.e., the set of $d$-dimensional unitary operators that conjugate $d$-dimensional generalized Pauli matrices to $d$-dimensional generalized Pauli matrices. If the dimension is clear from the context, we simply write $\calC$. 

We denote by $\delta_S$ the indicator function for a set $S$, with $\delta_S(x) = 1$ if $x \in S$ and $\delta_S(x) = 0$ otherwise. For a point $s$, it is understood that $\delta_s := \delta_{\{s\}}$.

For a string $x \in \{0,1\}^n$, we use $\abs{x} = n$ to denote the length of the string. By default, all operations on bitstrings are assumed to be performed over $\F_2$.

\subsection{Acknowledgements}
This material is based upon work supported by the U.S. Department of Energy, Office of Science, National Quantum Information Science Research Centers, Quantum Systems Accelerator.

We thank James Bartusek, Bhaskar Roberts, Mark Zhandry, Henry Yuen, and Fermi Ma for helpful conversations. We would also like to thank Prabhanjan Ananth for alerting us to the fact that we were missing citations of existing simultaneous Goldreich-Levin theorems in a previous version.

\section[Quantum State Indistinguishability Obfuscation (qsiO)]{Quantum State Indistinguishability Obfuscation ($\qsiO$)}
\label{sec:qsiO}
In this section, we define quantum state indistinguishability obfuscation ($\qsiO$). We show that $\qsiO$ achieves ``best-possible'' copy protection, and we describe a construction of $\qsiO$ relative to a quantum oracle.

\subsection{Definitions}
We start by defining a ``quantum implementation'' of a classical function.

\begin{definition}[Quantum implementation of a classical function]
\label{def:quantum-implementation}
Let $l_{in}, l_{out} \in \mathbb{N}$, $f: \{0,1\}^{l_{in}} \rightarrow \{0,1\}^{l_{out}}$, and $\epsilon \in [0,1]$.
A $(1-\epsilon)$-\emph{quantum implementation} of $f$ is a pair $(\rho, C)$, where $\rho$ is a state and $C$ is a quantum circuit, such that
\[
    \forall x \in \{0,1\}^{l_{in}}, \ \Pr[C(\rho,x) = f(x)] \geq 1-\epsilon \,.
\]
\end{definition}

For a quantum implementation $(\rho, C)$, we refer to its \emph{size} as the maximum between the number of qubits of $\rho$ and the number of gates of the circuit $C$. We now define $\qsiO$.

\begin{definition}[Quantum state indistinguishability obfuscator ($\qsiO$)]
Let $\{Q_\secpar\}_{\secpar \in \mathbb{N}}$ be a family of $(1-\negl)$-quantum implementations (of possibly different functions), where $\negl[]$ is some negligible function. A \emph{quantum state indistinguishability obfuscator} for $\{Q_\secpar\}_{\secpar \in \mathbb{N}}$ is a QPT algorithm $\qsiO$ that takes as input a security parameter $\secparam$, a quantum implementation $(\rho,C) \in Q_{\secpar}$, and outputs a pair $(\rho',C')$. Additionally, $\qsiO$ should satisfy the following. 
    \begin{itemize}
        \item (Correctness) There exists a negligible function $\negl[]'$ such that, for any $\secpar \in \mathbb{N}$, if $(\rho, C) \in Q_{\secpar}$ is a $(1-\negl)$-quantum implementation of some function $f$, then $\qsiO(\rho, C)$ is a $(1-\mathsf{negl}'(\lambda))$-quantum implementation of $f$.
        \item (Security) For any QPT distinguisher $D$, there exists a negligible function $\negl[]''$ such that the following holds. For all $\secpar$ and all pairs of $(1-\negl)$-quantum implementations $(\rho_0, C_0), (\rho_1, C_1) \in Q_{\secpar}$ of the same function $f$,
        \[
            \Big|\Pr[D(\secparam, \qsiO(\secparam,(\rho_0,C_0))) = 1] - \Pr[D(\secparam, \qsiO(\secparam,(\rho_1,C_1))) = 1]\Big| \le \negl[]''(\lambda) \,.
        \]
    \end{itemize}
\end{definition}

In this paper, we will make use of $\qsiO$ for all polynomial-size quantum implementations. That is, we will assume the existence of $\qsiO$ for $\{Q_{\secpar}\}_{\secpar \in \mathbb{N}}$, where $Q_{\secpar}$ is the set of $(1-\negl)$-quantum implementations of size at most $\secpar$, for some negligible function $\negl[]$.

For ease of notation, we will often omit writing $\secparam$ as an input to $\qsiO$. We will sometimes apply $\qsiO$ to a circuit $C$ \emph{without} auxiliary quantum input, or to a classical circuit $C$. In this case, we simply write $\qsiO(C)$. If the circuit $C$ is classical we sometimes identify it with the function $f$ that it is computing, and simply write $\qsiO(f)$.

\subsection{Best-possible copy protection}

It is not hard to see that $\qsiO$, as defined in the previous section,  achieves \emph{best-possible} copy protection. In this section, we state a definition of copy protection that is quite general, and encompasses all the variants that we will later consider in \Cref{sec:cp}. 

\begin{definition}[Copy protection, correctness]\label{def:cp-correctness}
Let $\calF = \{\calF_{\secpar}\}_{\secpar \in \mathbb{N}}$ be a family of classical circuits. A QPT algorithm $\CP$ is a copy protection scheme for $\mathcal{F}$ if the following holds, for some negligible function $\negl[]$:
\begin{itemize}
    \item $\CP$ takes as input a security parameter $\secparam$ and a circuit $f \in \mathcal{F}_{\secpar}$, and outputs a $(1-\negl)$-quantum implementation $(\rho,C)$ of $f$ with probability $1-\negl$.
\end{itemize}
 \end{definition}

The definition of security below is stated in terms of a circuit $\Ver$ that the challenger runs on each half of a state received from the adversary (the ``pirate''). Some readers may be more familiar with a security game where the challenger samples a pair of inputs to the copy-protected function $f$, and expects two parties Alice and Bob to return the value of $f$ at those inputs. The security game in \Cref{game:CP} subsumes such a security game (by taking $\Ver(f,A)$ to be the circuit that first samples an input $x$ to $f$, and then runs ``Alice's circuit'' $A^x$ on the input state). We elect to keep the definition general here, so as not to limit the applicability of our ``best-possible copy protection'' result (\Cref{theorem:best-possible}). Later, when we discuss copy protection of concrete functionalities in \Cref{sec:cp}, we opt for a more explicit description of the security game. 

\begin{definition}[Copy protection, security] \label{def:cp}
    Let $\calF = \{\calF_{\secpar}\}_{\secpar \in \mathbb{N}}$ be a family of classical circuits. Let $\CP$ be a copy protection scheme for $\mathcal{F}$.
    
    Let $\Ver = \{\Ver_{\secpar} \}_{\secpar \in \mathbb{N}}$ be a uniform family of polynomial-size quantum circuits, where $\Ver_{\secpar}$ takes as input a function $f \in \calF_{\lambda}$, a family of $\poly(\lambda)$-size quantum circuits $\{Q^x\}_{x \in \textnormal{Domain}(f)}$, and a quantum state, and outputs a single bit. Let $\delta: \mathbb{N} \rightarrow [0,1]$.
    
    We say that $\CP$ is $(\Ver, \delta)$-secure if, for all QPT algorithms $\Adv$, there exists a negligible function $\negl[]$ such that, for all $\secpar$,
\[
    \Pr[\emph{\CPExpt}_{\CP,\Adv,\Ver}(\secpar) = 1] \leq \delta(\secpar) + \negl \,,
\]
where $\emph{\CPExpt}_{\CP,\Adv,\Ver}$ is defined in \Cref{game:CP}.\footnote{Note that in this security game the function $f$ is sampled uniformly from $\mathcal{F}_{\lambda}$. This restriction is essentially without loss of generality, since one can pad the description of the circuits $f$ with additional bits that do not affect the circuit itself, but serve the purpose of changing the probability mass on a particular circuit.} For a function $f \in F_{\lambda}$ and a family of circuits $Q$, we use the notation $\Ver(f, Q) := \Ver(f,Q, \cdot)$ (so $\Ver(f, Q)$ denotes a quantum circuit that takes as input a state and outputs a single bit).
\end{definition}

\begin{figure}[H]
\pcb{
    \textbf{Challenger} \< \< \textbf{Adversary} \\[][\hline]
    \< \< \\[-0.5\baselineskip]
    f \from \calF_\secpar\<\< \\
    \sigma := \CP(\secparam, f) \< \< \\
    \< \sendmessageright*{\sigma} \< \\
    \<\< (A, B, \rho_{\calA,\calB}) \from \Adv(\sigma) \\
    \< \sendmessageleft*{A, B, \rho_{\calA,\calB}} \< \\
    \textnormal{Run } \Ver_{\secpar}(f,A) \otimes \Ver_{\secpar}(f,B) \,\rho_{\mathcal{A, B}} \\
    \textnormal{Output $1$ if both outcomes are 1}  
}
\caption{$\CPExpt_{\CP,\Adv,\Ver}(\secpar)$. The challenger samples $f \from \calF_\secpar$, creates the quantum implementation $\sigma = \CP(f)$, and sends it to the adversary. The adversary maps this to a state $\rho_{\calA\calB}$ on the two registers $\calA$, $\calB$, and sends $\rho_{\calA\calB}$ back to the challenger, along with (descriptions of) families of quantum circuits $A$ and $B$ on $\calA$ and $\calB$ respectively. The challenger runs $\Ver_{\secpar}(f,A) \otimes \Ver_{\secpar}(f,B)$ on $\rho_{\mathcal{A, B}}$, and outputs $1$ if both outcomes are $1$.} \label{game:CP}
\end{figure} 

\begin{theorem}[``Best-possible'' copy protection] \label{theorem:best-possible}
    Let $\calF = \{\calF_{\secpar}\}_{\secpar \in \mathbb{N}}$ be a family of classical circuits. Suppose there exists a copy protection scheme for $\calF = \{\calF_{\secpar}\}_{\secpar \in \mathbb{N}}$ that is $(\Ver, \delta)$-secure (for some $\Ver$, and $\delta$ as in Definition \ref{def:cp}). Let $\qsiO$ be a secure quantum state indistinguishability obfuscator for $\calF$. Then, $\qsiO$ is a $(\Ver, \delta)$-secure copy protection scheme for $\calF$.
\end{theorem}
\begin{proof}
    Let $\CP$ be the $(\Ver,\delta)$-secure copy protection scheme for $\calF$ that exists by hypothesis. Let $\Adv$ be any efficient adversary for $\CPExpt$. Since the challenger in $\CPExpt$ is also efficient, we have by the security guarantee of $\qsiO$, that there exists a negligible function $\negl[]$ such that, for all $\lambda$,
    \begin{equation} \label{eq:best-possible-eq1}
        \Pr[\CPExpt_{\qsiO,\Adv,\Ver}(\secpar) = 1] \le \Pr[\CPExpt_{\qsiO \circ \CP,\Adv,\Ver}(\secpar) = 1] + \negl.
    \end{equation}
    Consider a reduction $\Red[\Adv]$ for $\CPExpt$ that simply applies $\qsiO$ before forwarding $\sigma$ to $\Adv$, and then forwards the response from $\Adv$ to the challenger. Formally, $\Red$ is defined by the following behavior in $\CPExpt$.
    
    \indent $\CPExpt_{\CP,\Red[\Adv],\Ver}(\secpar)$:
    \begin{enumerate}
        \item The challenger samples $f \from \calF_\secpar$, computes $\sigma = \CP(f)$, and sends $\sigma$ to $\Red$.
        \item $\Red$ computes $\sigma' = \qsiO(\sigma) = (\qsiO \circ \CP)(f)$ and sends $\sigma'$ to $\Adv$.
        \item $\Adv$ sends $(A,B, \rho_{\calA\calB})$ to $\Red$, which forwards this to the challenger.
        \item The challenger then runs $\Ver_{\secpar}(f,A) \otimes \Ver_{\secpar}(f,B) \,\rho_{\mathcal{A, B}}$ and outputs $1$ if both outcomes are $1$.
    \end{enumerate}
    By construction,
    \begin{equation} \label{eq:best-possible-eq2}
        \Pr[\CPExpt_{\qsiO \circ \CP,\Adv,\Ver}(\secpar) = 1] = \Pr[\CPExpt_{\CP,\Red[\Adv],\Ver}(\secpar) = 1].
    \end{equation}
    Finally, the assumption that $\CP$ is a $(\Ver,\delta)$-secure copy protection scheme for $\calF$ implies that
    \begin{equation} \label{eq:best-possible-eq3}
        \Pr[\CPExpt_{\CP,\Red[\Adv],\Ver}(\secpar) = 1] = \negl.
    \end{equation}
    Combining Equations \eqref{eq:best-possible-eq1}, \eqref{eq:best-possible-eq2}, and \eqref{eq:best-possible-eq3} gives the result.
\end{proof}

\begin{remark}
    In this work we only consider $\qsiO$ for quantum implementations of deterministic functions. It would be interesting to explore an extended definition that allows for quantum implementations of \emph{randomized} functions. It is plausible that a proper formalization would yield best-possible one time programs in a similar way to \Cref{theorem:best-possible}.
\end{remark}

\subsection[Constructing qsiO]{Constructing $\qsiO$}
\label{subsec:construction}
In this section, we give a construction $\Obf$ of $\qsiO$ relative to a quantum oracle. Before describing it formally, we give an informal description:
\begin{itemize}
    \item $\Obf$ takes as input a quantum implementation $(\rho, \mathsf{Eval})$ of some function $f$, where $\Eval$ is assumed to be a universal evaluation circuit without loss of generality.
    \item $\Obf$ samples a uniformly random Clifford unitary $C$ and outputs the state $\tilde{\rho} = C \rho C^{\dagger}$, alongside an oracle implementing the unitary $G_C = C \mathsf{Eval} C^{\dagger}$ (where $\mathsf{Eval}$ here refers to the unitary part of the evaluation circuit).
\end{itemize} 
In other words, $\qsiO$ applies a Clifford one-time pad to ``hide'' the input state $\rho$; the oracle $G_C$ undoes the one-time pad, evaluates the function $f$, and then re-applies the one-time pad. This allows a user to evaluate $f$, while intuitively keeping the state $\rho$ hidden at all times.

The main tool in our proof is the ``Clifford twirl'' \cite{ABOEM17}, which would already suffice if the adversary were only allowed to make a \emph{single} query to $G_C$. However, the adversary can make any polynomial number of queries, so a more careful argument is required. Our argument additionally makes use of a recently-introduced tool called the ``admissible oracle lemma'' \cite{GJMZ23}, which allows us to reduce the security of the many-query game to the security of the one-query game.

\begin{const}
\label{const:qsiO}
$\Obf$ takes as input a quantum implementation $(\rho, \Eval)$ of some function $f$, where $\rho$ is a state on a register $\calA$. We assume without loss of generality that the circuit $\Eval$ consists of a unitary on $\calA$ as well as an input register $\mathcal{X}$ and  an output register $\mathcal{Y}$, followed by a measurement of register $\mathcal{Y}$. For ease of notation, we will identify the algorithm $\Eval$ (which includes a measurement) with its unitary part when it is clear from the context. We assume without loss of generality that the unitary $\Eval$ uncomputes all of its intermediate steps, leaving the result on $\mathcal{Y}$.
 
$\Obf(\rho, U)$ proceeds as follows:
 
\begin{itemize}
\item Sample $C \leftarrow \calC$, where $\calC$ is the Clifford group. Let $\tilde{\rho} = C \rho C^{\dagger}$.
\item Let $G_C$ be the unitary acting on registers $\calA, \mathcal{X}, \mathcal{Y}$ defined as $G_C= (C_{\calA} \otimes I_{\mathcal{X}\mathcal{Y}}) \Eval (C^{\dagger}_{\calA}\otimes I_{\mathcal{X}\mathcal{Y}})$.
\item Let $\tilde{U}^{G_C}$ be the quantum circuit, with oracle access to $G_C$, that behaves as follows: On input $\rho$, run $G_C$ on input $\rho_{\calA} \otimes \ketbra{x}_{\mathcal{X}} \otimes \ketbra{0}_{\mathcal{Y}}$; measure register $\mathcal{Y}$ and output the outcome.
\item Output $(\tilde{\rho}, \,\tilde{U}^{G_C})$ (since this is an oracle construction, what we mean is that the algorithm $\Obf$ outputs the description of the oracle algorithm $\tilde{U}$, and the oracle $G_C$ is publicly available).
\end{itemize}
\end{const}

We show that Construction \ref{const:qsiO} is $\qsiO$ in a model where the adversary has access to the oracle $G_C$.

\begin{theorem} \label{theorem:construction}
Construction \ref{const:qsiO} is $\qsiO$.
\end{theorem}
\begin{proof}
We prove security via three hybrids. The first hybrid corresponds to the original $\qsiO$ security game. The second corresponds to a ``purified'' version of the $\qsiO$ game, which is easily seen to be equivalent to the original. The third hybrid is identical to the second, except that the adversary has access to a different oracle: This new oracle does not evaluate the function unless the register containing $C$ is in uniform superposition.

Finally, we show that the distinguishing advantage in the third hybrid is \emph{zero} by invoking the ``admissible oracle lemma'' of \cite{GJMZ23}.

\vspace{2mm}
\noindent \textit{Hybrid 1}: The original $\qsiO$ security game. 

\vspace{2mm}
\noindent \textit{Hybrid 2}: A ``purified'' version of the $\qsiO$ game. Let $(\ket{\psi_1}, U_1)$ and $(\ket{\psi_2}, U_2)$ be two quantum implementations of the same classical functionality $f: \{0,1\}^{l_{in}} \rightarrow \{0,1\}^{l_{out}}$, where $\ket{\psi_1}$ and $\ket{\psi_2}$ are states on some register $\calB_1$, and $U_1, U_2$ are unitaries on $\calB_1$ and some other register $\calR$. For the rest of the proof we assume, for simplicity and without loss of generality, that the unitaries $U_1$ and $U_2$ are equal to a fixed universal unitary $\Eval_{\calB_1, \calR}$.

Let $\calA,\calB_1, \calB_2,\calR$ be registers, and let $\Pi'$ be the subspace spanned by all the states of the form
\begin{equation}
\label{eq:1}
\ket{C}_{\calA} \otimes C(\ket*{\psi}_{\calB_1} \otimes \ket*{0^\secpar}_{\calB_2}) \otimes \ket{x, y}_{\calR} \,, 
\end{equation} 
where $C$ is any Clifford unitary on register $\calB := \calB_1\calB_2$, $\ket{\psi}$ is any state on $\calB_1$, and $(x,y) \in \{0,1\}^{l_{in}} \times \{0,1\}^{l_{out}}$.

The unitary $G'$ acts as identity on the orthogonal complement of $\Pi'$, and as follows on $\Pi'$:
\[
    G' \Pi' = \sum_{C \in \calC} \ketbra{C}_{\calA} \otimes (C_{\calB} \otimes I_{\calR}) (\Eval_{\calB_1,\calR} \otimes I_{\calB_2}) (C^\dagger_{\calB} \otimes I_{\calR}).
\]
In the rest of the section, when it is clear from the context, we omit writing tensor products with identities, e.g.\ we write 
\[
    G' \Pi' = \sum_{C \in \calC} \ketbra{C}_{\calA} \otimes C_{\calB} (\Eval_{\calB_1,\calR}) C^\dagger_{\calB}.
\]

The game is as follows:
\begin{enumerate}
\item The challenger samples $b \leftarrow \{0,1\}$. Then, it creates the state $\abs{\calC}^{-1/2} \sum_{C \in \calC} \ket*{C}_{\calA} \otimes C(\ket*{\psi_b} \otimes \ket*{0^\secpar})_{\calB}$.

\item The adversary receives register $\calB$ from the challenger, as well as query access to the oracle $G$ (where the adversary controls $\calR$). The adversary returns a guess $b' \in \{0,1\}$. 
\end{enumerate}
The adversary wins if $b' = b$.

\vspace{2mm}
\noindent \textit{Hybrid 3}: Identical to Hybrid 2, except the adversary has access to a different oracle $G$ defined as follows. 
\begin{itemize}
\item Let $\Pi$ be the subspace spanned by all the states of the form
\begin{equation}
\label{eq:2}
\abs{\calC}^{-1/2} \sum_{C \in \calC} \ket*{C} \otimes C(\ket*{\psi} \otimes \ket*{0^\secpar}) \otimes \ket{x,y} \,,
\end{equation}
where $\ket{\psi}$ is any state on $\calB_1$ and $(x,y) \in \{0,1\}^{l_{in}} \times \{0,1\}^{l_{out}}$. The unitary $G$ acts as identity on the orthogonal complement of $\Pi$, and as follows on $\Pi$:
\[
G \Pi = \sum_{C \in \calC} \ketbra{C} \otimes C (\Eval) C^\dagger
\]
\end{itemize}

\vspace{2mm}

We first show that the adversary's advantage in Hybrid 1 and Hybrid 2 is identical.

\begin{lemma}
For any adversary $A$,
$$ \Pr[A \textnormal{ wins in Hybrid } 1] = \Pr[A \textnormal{ wins in Hybrid } 2]  \,.$$
\end{lemma}
\begin{proof}
This is immediate since Hybrid 2 is just a purification of Hybrid 1.
\end{proof}

\begin{lemma}
\label{lem:4}
    For any adversary $A$ for Hybrids 2 and 3, there exists a negligible function $\negl[]$ such that, for all $\secpar$,
$$ \big| \Pr[A \textnormal{ wins in Hybrid } 2] - \Pr[A \textnormal{ wins in Hybrid } 3] \big| \leq \negl \,.$$
\end{lemma}

\begin{proof}
Let $W_{\calC} = \sum_{C \in \calC} \ketbra{C} \otimes C$. Define 
\begin{equation*}
\Eval_1 := \Eval \cdot \big(I_{\calA \calB_1 \calR} \otimes \ketbra{0^{\secpar}}_{\calB_2} \big) + I_{\calA \calB_1 \calR} \otimes (I -\ketbra{0^{\secpar}}_{\calB_2}) \,.
\end{equation*}
Then, notice that we can write $G'$ (from Hybrid 2) as
\begin{equation}
\label{eq:g1}
G' = W_{\calC} \Eval_1 (W_{\calC})^{\dagger} \,.
\end{equation}

Let $\ket{\tau} = |\calC|^{-1/2} \sum_{C \in \calC} \ket{C}$, and define 
\begin{equation*}
\Eval_2 := \Eval \cdot   \big(\ketbra{\tau}_{\calA} \otimes I_{\calB_1 \calR} \otimes \ketbra{0^{\secpar}}_{\calB_2} \big) + (I-\ketbra{\tau}_{\calA}) \otimes I_{\calB_1 \calR} \otimes \ketbra{0^{\secpar}}_{\calB_2}+ I_{\calA \calB_1 \calR} \otimes (I -\ketbra{0^{\secpar}}_{\calB_2}) \,.
\end{equation*}
We can write $G$ (from Hybrid 3) as
\begin{equation}
G = W_{\calC} \Eval_2 (W_{\calC})^{\dagger} \,. \label{eq:g2}
\end{equation}

Let $A$ be an adversary for Hybrid 2 and 3. Recall that Hybrids 2 and 3 are identical except that the oracle is $G'$ in Hybrid 2 and $G$ in Hybrid 3. Recall that the challenger initializes registers $\calA, \calB$ in the state 
\begin{equation}
\label{eq:psi0}
\ket{\Psi_0} := \abs{\calC}^{-1/2} \sum_{C \in \calC} \ket*{C}_{\calA} \otimes C(\ket*{\psi_b} \otimes \ket*{0^\secpar})_{\calB} \,.
\end{equation}
for some $b \in \{0,1\}$
Then $A$ then receives register $\calB$. Let $\calR$ denote the register where $\Eval$ writes its output, and let $\mathcal{Z}$ denote an additional work register used by $A$ (which includes an output register). Let $q$ be the number of queries to the oracle ($G'$ or $G$) made by $A$. Without loss of generality, for some unitary $U$ on $\calB \calR \mathcal{Z}$, we have that $A$ applies the sequence of unitaries $(G'U)^q$ in Hybrid 2, and $(GU)^q$ in Hybrid 3.

We will prove \Cref{lem:5}, which implies that $A$'s success probability in Hybrids 2 and 3 is negligibly close, as long as $q = \poly(\secpar)$. This will complete the proof of \Cref{lem:4}. 
\end{proof}

\begin{lemma}
\label{lem:5}
Let $\ket{\Psi_0}_{\calA\calB}$ be as defined in \eqref{eq:psi0} (note that this state depends on $\secpar$). Let $\ket{\phi}_{\mathcal{R}\mathcal{Z}}$ be any state, and $U$ any unitary (both of which implicitly depend on $\secpar$), and let $q \in \mathbb{N}$. Then, for all $\secpar$,
\begin{equation}
\|(G'U)^q \ket{\Psi_0} \otimes \ket{\phi} - (GU)^q \ket{\Psi_0} \otimes \ket{\phi}\| = q (q+1)2^{-\secpar} \,.
\end{equation}
\end{lemma}
\begin{proof} For convenience, we use the following notation throughout this proof: for $\epsilon >0$ and states $\ket{u}$ and $\ket{v}$, we write $\ket{u} \equiv_{\epsilon} \ket{v}$ as a shorthand for $\|\ket{u} - \ket{v} \| \leq \epsilon$. Let $m$ be the number of qubits in register $\calB_1$.
We prove \Cref{lem:5} by induction on the number of queries. Precisely, we show that, for all $i \in \{0,\ldots,q\}$:
\begin{itemize}
    \item[(i)] There exist unnormalized states $\ket{\phi_{x,z}}_{\mathcal{R}\mathcal{Z}}$ for $x,z \in \{0,1\}^{m+\secpar}$ such that, for all $\secpar$, $$(G'U)^i\ket{\Psi_0} \equiv_{i \cdot 2^{-\secpar}} \abs{\calC}^{-1/2} \sum_{C \in \calC} \sum_{x,z \in \{0,1\}^{m+\secpar}} \ket*{C}_{\calA} \otimes X^x Z^z C(\ket*{\psi_b} \otimes \ket*{0^\secpar})_{\calB} \otimes  \ket{\phi_{x,z}}\,.$$
    \item[(ii)] For all $\secpar$, $$(G'U)^i \ket{\Psi_0} \otimes \ket{\phi} \equiv_{i(i+1)\cdot 2^{-\secpar}} (GU)^i\ket{\Psi_0}\otimes \ket{\phi} \,.$$
\end{itemize} 
Clearly, both (i) and (ii) hold for $i=0$. Now, suppose (i) and (ii) hold for some $i$. We show that they both hold also for $i+1$. By the inductive hypothesis, we have
$$(G'U)^i\ket{\Psi_0} \otimes \ket{\phi} \equiv_{i \cdot 2^{-\secpar}} \abs{\calC}^{-1/2} \sum_{C \in \calC} \sum_{x,z \in \{0,1\}^{m+\secpar}} \ket*{C}_{\calA} \otimes X^x Z^z C(\ket*{\psi_b} \otimes \ket*{0^\secpar})_{\calB} \otimes  \ket{\phi_{x,z}}\,,$$
for some unnormalized states $\ket{\phi_{x,y}}_{\mathcal{R}\mathcal{Z}}$ for $x,z \in \{0,1\}^{m+\secpar}$.
Then, we have
\begin{align}
\label{eq:5}
(W_{\calC})^{\dagger} U (G'U)^i\ket{\Psi_0} \otimes \ket{\phi} \equiv_{i \cdot 2^{-\secpar}} (W_{\calC})^{\dagger} U \,\abs{\calC}^{-1/2} \sum_{C \in \calC} \sum_{x,z \in \{0,1\}^{m+\secpar}} \ket*{C}_{\calA} \otimes X^xZ^z C(\ket*{\psi_b} \otimes \ket*{0^\secpar})_{\calB} \otimes  \ket{\phi_{x,z}} \,.
\end{align}
By expanding the $\calB$ register in the Pauli basis, we can write $U = \sum_{x,y \in \{0,1\}^{m+\secpar}} X^x Z^z \otimes U_{xz}$ for some operators $U_{xz}$. Then, plugging this into \eqref{eq:5} we get
\begin{align}
(W_{\calC})^{\dagger} U (G'U)^i\ket{\Psi_0} \otimes \ket{\phi} &\equiv_{i \cdot 2^{-\secpar}} (W_{\calC})^{\dagger} \,\abs{\calC}^{-1/2} \sum_{x',z'} \sum_{C \in \calC} \sum_{x,z} \ket*{C}_{\calA} \otimes X^{x'}Z^{z'} X^xZ^z C(\ket*{\psi_b} \otimes \ket*{0^\secpar})_{\calB} \otimes  U_{x',z'}\ket{\phi_{x,z}} \nonumber \\
&= (W_{\calC})^{\dagger} \,\abs{\calC}^{-1/2} \sum_{C \in \calC} \sum_{x,z, x',z'} \ket*{C}_{\calA} \otimes (-1)^{z' \cdot x} X^{x+x'}Z^{z+z'} C(\ket*{\psi_b} \otimes \ket*{0^\secpar})_{\calB} \otimes  U_{x',z'}\ket{\phi_{x,z}} \nonumber\\
&= \abs{\calC}^{-1/2} \sum_{C \in \calC} \sum_{x,z, x',z'} \ket*{C}_{\calA} \otimes (-1)^{z' \cdot x} C^{\dagger} X^{x+x'}Z^{z+z'} C(\ket*{\psi_b} \otimes \ket*{0^\secpar})_{\calB} \otimes  U_{x',z'}\ket{\phi_{x,z}} \nonumber\\
&= \abs{\calC}^{-1/2} \sum_{C \in \calC} \ket*{C}_{\calA} \otimes (\ket*{\psi_b} \otimes \ket*{0^\secpar})_{\calB} \otimes  |\tilde{\phi}_{0,0} \rangle \nonumber\\
&+ \abs{\calC}^{-1/2} \sum_{C \in \calC} \ket*{C}_{\calA} \otimes \sum_{(x,z)\neq (0^{m+\secpar}, 0^{m+\secpar}) } C^{\dagger} X^{x}Z^{z} C(\ket*{\psi_b} \otimes \ket*{0^\secpar})_{\calB} \otimes  |\tilde{\phi}_{x,z}\rangle \,, \label{eq:88}
\end{align}
for some unnormalized states $|\tilde{\phi}_{x,z} \rangle$.
We will show that the second summand in the last expression has exponentially small weight on states such that the state on register $\calB_2$ is $\ket{0}^{\secpar}$. Precisely, we will show that
\begin{equation}
\label{eq:6}
\Big\|I_{\calA \calB_1 \calR} \otimes \ketbra{0^{\secpar}}_{\calB_2}\,\,\abs{\calC}^{-1/2} \sum_{C \in \calC} \ket*{C}_{\calA} \otimes \sum_{(x,z)\neq (0^{m+\secpar}, 0^{m+\secpar}) } C^{\dagger} X^{x}Z^{z} C(\ket*{\psi_b} \otimes \ket*{0^\secpar})_{\calB} \otimes  |\tilde{\phi}_{x,z}\rangle \Big\| \leq 2^{-\secpar}  \,.  
\end{equation}
We will prove \eqref{eq:6} at the end. Now, recall that a Clifford operator is uniquely specified by the fact that it maps Pauli operators to Pauli operators when acting by conjugation. For $C \in \calC$ and $x, z \in \{0,1\}^{m+\secpar}$, let $\pi^X_C(x,z) \in \{0,1\}^{\secpar}$ 
be defined such that
$$(X^{\tilde{x}_1} \otimes X^{\pi^X_C(x,z)} ) Z^{\tilde{z}} = C^{\dagger} X^x Z^z C \,. $$
for some $\tilde{x}_1 \in \{0,1\}^m, \tilde{z} \in \{0,1\}^{m+\secpar}$. In other words, $\pi^X_C(x,z)$ corresponds to the \emph{last $\secpar$ bits} of the Pauli $X$ string obtained by conjugating $X^xZ^z$ by $C$.

Assuming \eqref{eq:6} is true, we have 
\begin{equation}
\label{eq:7}
   \Big\| \abs{\calC}^{-1/2} \sum_{C \in \calC} \ket*{C}_{\calA} \otimes \sum_{\substack{(x,z)\neq (0^{m+\secpar}, 0^{m+\secpar}): \\ \pi_C^X(x,z) = 0^{\secpar}}} C^{\dagger} X^{x}Z^{z} C(\ket*{\psi_b} \otimes \ket*{0^\secpar})_{\calB} \otimes  |\tilde{\phi}_{x,z}\rangle \Big\| \leq 2^{-\secpar} \,.
\end{equation}
Then, we have
\begin{align}
&\Eval_1 (W_{\calC})^{\dagger} U (G'U)^i\ket{\Psi_0} \otimes \ket{\phi} \nonumber \\&\equiv_{i \cdot 2^{-\secpar}+2^{-\secpar}} \Eval_1  \bigg( \abs{\calC}^{-1/2} \sum_{C \in \calC} \ket*{C}_{\calA} \otimes (\ket*{\psi_b} \otimes \ket*{0^\secpar})_{\calB} \otimes  |\tilde{\phi}_{0,0} \rangle \nonumber\\ 
&\quad \quad \quad \quad \quad \quad \quad + \abs{\calC}^{-1/2} \sum_{C \in \calC} \ket*{C}_{\calA} \otimes \sum_{\substack{(x,z)\neq (0^{m+\secpar}, 0^{m+\secpar}): \\ \pi_C^X(x,z) \neq 0^{\secpar}}} C^{\dagger} X^{x}Z^{z} C(\ket*{\psi_b} \otimes \ket*{0^\secpar})_{\calB} \otimes  |\tilde{\phi}_{x,z}\rangle \bigg) \label{eq:99}\\
&= \Eval \,\abs{\calC}^{-1/2} \sum_{C \in \calC} \ket*{C}_{\calA} \otimes (\ket*{\psi_b} \otimes \ket*{0^\secpar})_{\calB} \otimes  |\tilde{\phi}_{0,0} \rangle \nonumber\\
&\quad+ \abs{\calC}^{-1/2} \sum_{C \in \calC} \ket*{C}_{\calA} \otimes \sum_{\substack{(x,z)\neq (0^{m+\secpar}, 0^{m+\secpar}): \\ \pi_C^X(x,z) \neq 0^{\secpar}}} C^{\dagger} X^{x}Z^{z} C(\ket*{\psi_b} \otimes \ket*{0^\secpar})_{\calB} \otimes  |\tilde{\phi}_{x,z}\rangle \,, \label{eq:100}
\end{align}
where \eqref{eq:99} follows from \eqref{eq:88}, \eqref{eq:7}, unitarity of $\Eval_1$, and a triangle inequality; and
\eqref{eq:100} follows from the definition of $\Eval_1$.

Notice that, by definition of $\Eval$, $$\Eval (\ket*{\psi_b} \otimes \ket*{0^\secpar})_{\calB} \otimes  |\tilde{\phi}_{0,0} \rangle = (\ket*{\psi_b} \otimes \ket*{0^\secpar})_{\calB} \otimes  |\tilde{\tilde{\phi}}_{0,0} \rangle \,,$$
for some other state $|\tilde{\tilde{\phi}}_{0,0} \rangle$. So, plugging this into \eqref{eq:100}, gives
\begin{align}
\Eval_1 (W_{\calC})^{\dagger} U (G'U)^i\ket{\Psi_0} \otimes \ket{\phi} &\equiv_{(i+1) \cdot 2^{-\secpar}} \abs{\calC}^{-1/2} \sum_{C \in \calC} \ket*{C}_{\calA} \otimes (\ket*{\psi_b} \otimes \ket*{0^\secpar})_{\calB} \otimes  |\tilde{\tilde{\phi}}_{0,0} \rangle \nonumber\\
&\quad+ \abs{\calC}^{-1/2} \sum_{C \in \calC} \ket*{C}_{\calA} \otimes \sum_{\substack{(x,z)\neq (0^{m+\secpar}, 0^{m+\secpar}): \\ \pi_C^X(x,z) \neq 0^{\secpar}}} C^{\dagger} X^{x}Z^{z} C(\ket*{\psi_b} \otimes \ket*{0^\secpar})_{\calB} \otimes  |\tilde{\phi}_{x,z}\rangle \nonumber \,.
\end{align}
Applying $W_{\calC}$ to both sides, we get, by the unitarity of $W_{\calC}$ and using its definition,
\begin{align}
W_{\calC} \Eval_1 (W_{\calC})^{\dagger} U (G'U)^i\ket{\Psi_0} \otimes \ket{\phi} & \equiv_{(i+1) \cdot 2^{-\secpar}}  \,\abs{\calC}^{-1/2} \sum_{C \in \calC} \ket*{C}_{\calA} \otimes C (\ket*{\psi_b} \otimes \ket*{0^\secpar})_{\calB} \otimes  |\tilde{\tilde{\phi}}_{0,0} \rangle \nonumber \\
&\quad+ \abs{\calC}^{-1/2} \sum_{C \in \calC} \ket*{C}_{\calA} \otimes \sum_{\substack{(x,z)\neq (0^{m+\secpar}, 0^{m+\secpar}): \\ \pi_C^X(x,z) \neq 0^{\secpar}}}  X^{x}Z^{z} C(\ket*{\psi_b} \otimes \ket*{0^\secpar})_{\calB} \otimes  |\tilde{\phi}_{x,z}\rangle \,. \nonumber
\end{align}
Plugging $G' = W_{\calC} \Eval_1 (W_{\calC}^{\dagger})$ into the above gives
\begin{align}
(G'U)^{i+1} \ket{\Psi_0} \otimes \ket{\phi} & \equiv_{(i+1) \cdot 2^{-\secpar}}  \,\abs{\calC}^{-1/2} \sum_{C \in \calC} \ket*{C}_{\calA} \otimes C (\ket*{\psi_b} \otimes \ket*{0^\secpar})_{\calB} \otimes  |\tilde{\tilde{\phi}}_{0,0} \rangle \nonumber \\
&\quad+ \abs{\calC}^{-1/2} \sum_{C \in \calC} \ket*{C}_{\calA} \otimes \sum_{\substack{(x,z)\neq (0^{m+\secpar}, 0^{m+\secpar}): \\ \pi_C^X(x,z) \neq 0^{\secpar}}}  X^{x}Z^{z} C(\ket*{\psi_b} \otimes \ket*{0^\secpar})_{\calB} \otimes  |\tilde{\phi}_{x,z}\rangle \,,
\end{align}
which establishes item (i) of the inductive step. Next, we establish item (ii). From \eqref{eq:100}, we know that
\begin{align}
&\Eval_1 (W_{\calC})^{\dagger} U (G'U)^i\ket{\Psi_0} \otimes \ket{\phi} \nonumber \\&\equiv_{(i+1) \cdot 2^{-\secpar}}  \Eval \,\abs{\calC}^{-1/2} \sum_{C \in \calC} \ket*{C}_{\calA} \otimes (\ket*{\psi_b} \otimes \ket*{0^\secpar})_{\calB} \otimes  |\tilde{\phi}_{0,0} \rangle \nonumber\\
&\quad+ \abs{\calC}^{-1/2} \sum_{C \in \calC} \ket*{C}_{\calA} \otimes \sum_{\substack{(x,z)\neq (0^{m+\secpar}, 0^{m+\secpar}): \\ \pi_C^X(x,z) \neq 0^{\secpar}}} C^{\dagger} X^{x}Z^{z} C(\ket*{\psi_b} \otimes \ket*{0^\secpar})_{\calB} \otimes  |\tilde{\phi}_{x,z}\rangle \,. \label{eq:101}
\end{align}
Then, we have
\begin{align}
\textnormal{RHS of }\eqref{eq:101} &= \Eval \big(\ketbra{\tau}_{\calA} \otimes I_{\calB_1 \calR} \otimes \ketbra{0^{\secpar}}_{\calB_2} \big) \,\abs{\calC}^{-1/2} \sum_{C \in \calC} \ket*{C}_{\calA} \otimes (\ket*{\psi_b} \otimes \ket*{0^\secpar})_{\calB} \otimes  |\tilde{\phi}_{0,0} \rangle \nonumber\\
&\quad+ \abs{\calC}^{-1/2} \sum_{C \in \calC} \ket*{C}_{\calA} \otimes \sum_{\substack{(x,z)\neq (0^{m+\secpar}, 0^{m+\secpar}): \\ \pi_C^X(x,z) \neq 0^{\secpar}}} C^{\dagger} X^{x}Z^{z} C(\ket*{\psi_b} \otimes \ket*{0^\secpar})_{\calB} \otimes  |\tilde{\phi}_{x,z}\rangle \label{eq:11} \\
&= \Eval_2 \bigg( \abs{\calC}^{-1/2} \sum_{C \in \calC} \ket*{C}_{\calA} \otimes (\ket*{\psi_b} \otimes \ket*{0^\secpar})_{\calB} \otimes  |\tilde{\phi}_{0,0} \rangle \nonumber\\ 
&\quad \quad \quad \quad \quad \quad \quad + \abs{\calC}^{-1/2} \sum_{C \in \calC} \ket*{C}_{\calA} \otimes \sum_{\substack{(x,z)\neq (0^{m+\secpar}, 0^{m+\secpar}): \\ \pi_C^X(x,z) \neq 0^{\secpar}}} C^{\dagger} X^{x}Z^{z} C(\ket*{\psi_b} \otimes \ket*{0^\secpar})_{\calB} \otimes  |\tilde{\phi}_{x,z}\rangle \bigg) \label{eq:12} \\
&\equiv_{(i+1) \cdot 2^{-\secpar}} \Eval_2 (W_{\calC})^{\dagger} U (G'U)^i\ket{\Psi_0} \otimes \ket{\phi} \,, \label{eq:133}
\end{align}
where \eqref{eq:11} follows from the definition of $\ket{\tau}$; \eqref{eq:12}  from the definition of $\Eval_2$; and \eqref{eq:133} follows from the identical reasons as \eqref{eq:99}.

Overall, this gives, by a triangle inequality, 
\begin{equation}
    \label{eq:13}
\Eval_1 (W_{\calC})^{\dagger} U (G'U)^i\ket{\Psi_0} \otimes \ket{\phi} \equiv_{2 (i+1) \cdot 2^{-\secpar}}  \Eval_2 (W_{\calC})^{\dagger} U (G'U)^i\ket{\Psi_0} \otimes \ket{\phi}\,.
\end{equation}

Hence, we have
\begin{align}
(G'U)^{i+1}\ket{\Psi_0} \otimes \ket{\phi} &= G'U (G'U)^{i}\ket{\Psi_0} \otimes \ket{\phi} \nonumber \\ 
&= W_{\calC} \Eval_1 W_{\calC}^{\dagger} U (G'U)^i\ket{\Psi_0} \otimes \ket{\phi} \label{eq:15} \\
&\equiv_{2 (i+1) \cdot 2^{-\secpar}} W_{\calC} \Eval_2 W_{\calC}^{\dagger} U (G'U)^i\ket{\Psi_0} \otimes \ket{\phi} \label{eq:16} \\
&= GU (G'U)^i\ket{\Psi_0} \otimes \ket{\phi} \label{eq:17}\\
&\equiv_{i \cdot (i+1) \cdot 2^{-\secpar}} GU (GU)^i\ket{\Psi_0} \otimes \ket{\phi} \label{eq:18}\\
&= (GU)^{i+1} \ket{\Psi_0} \otimes \ket{\phi} \,, \nonumber
\end{align}
where \eqref{eq:15} is due to \eqref{eq:g1}; \eqref{eq:16} is due to \eqref{eq:13}; \eqref{eq:17} is due to \eqref{eq:g2}; and \eqref{eq:18} is by the inductive hypothesis. Overall, by a triangle inequality, we have
\begin{equation*}
(G'U)^{i+1}\ket{\Psi_0} \otimes \ket{\phi} \equiv_{2 (i+1) \cdot 2^{-\secpar} + i \cdot (i+1) \cdot 2^{-\secpar}} (GU)^{i+1} \ket{\Psi_0} \otimes \ket{\phi} \,,
\end{equation*}
which simplifies to
\begin{equation*}
(G'U)^{i+1}\ket{\Psi_0} \otimes \ket{\phi} \equiv_{(i+1)(i+2)\cdot 2^{\secpar}} (GU)^{i+1} \ket{\Psi_0} \otimes \ket{\phi} \,,
\end{equation*}
This establishes exactly item (ii) of the inductive hypothesis, and hence \Cref{lem:5} (assuming Equation \eqref{eq:6}).

To conclude the proof of \Cref{lem:5}, we are left with proving Equation \eqref{eq:6}. We will make use of the ``Clifford twirl.''
\begin{lemma}[Clifford twirl \cite{ABOEM17}]
\label{lem:clifford-twirl}
Let $n\in\mathbb{N}$. Let $\ket{\Psi}$ be any state on $n$ qubits. Let $x,z, x',z' \in \{0,1\}^n$ such that $(x,z) \neq (x',z')$. Let $\calC$ be the Clifford group on $n$ qubits. Then, 
$$\sum_{C \in \calC} C^{\dagger}X^{x'}Z^{z'} C \ketbra{\Psi} C^{\dagger}X^xZ^z C = 0 \,.$$
\end{lemma}

We have
\begin{align}
&\Big\|I_{\calA \calB_1 \calR} \otimes \ketbra{0^{\secpar}}_{\calB_2}\,\,\abs{\calC}^{-1/2} \sum_{C \in \calC} \ket*{C}_{\calA} \otimes \sum_{(x,z)\neq (0^{m+\secpar}, 0^{m+\secpar}) } C^{\dagger} X^{x}Z^{z} C(\ket*{\psi_b} \otimes \ket*{0^\secpar})_{\calB} \otimes  |\tilde{\phi}_{x,z}\rangle \Big\|^2 \nonumber \\
&= \frac{1}{|\calC|} \sum_{C \in \calC} \big\|(I \otimes \ketbra{0^{\secpar}}) \sum_{(x,z)\neq (0^{m+\secpar}, 0^{m+\secpar})} C^{\dagger} X^{x}Z^{z} C(\ket*{\psi_b} \otimes \ket*{0^\secpar})_{\calB} \otimes  |\tilde{\phi}_{x,z}\rangle \big\|^2 \nonumber\\
&=\frac{1}{|\calC|} \sum_{C \in \calC} \sum_{(x,z)\neq (0^{m+\secpar}, 0^{m+\secpar})} \big\|(I \otimes \ketbra{0^{\secpar}})  C^{\dagger} X^{x}Z^{z} C(\ket*{\psi_b} \otimes \ket*{0^\secpar})_{\calB} \otimes  |\tilde{\phi}_{x,z}\rangle  \big\|^2 \nonumber \\
&+ \frac{1}{|\calC|} \sum_{C \in \calC} \sum_{(x,z) \neq (x',z') \neq (0^{m+\secpar}, 0^{m+\secpar})} \bra{\psi}\bra{0^{\secpar}}  C^{\dagger}X^xZ^z C (I\otimes \ketbra{0^{\secpar}})C^{\dagger}X^{x'}Z^{z'} C\ket{\psi}\ket{0^{\secpar}} \cdot \langle \tilde{\phi}_{x,z} | \tilde{\phi}_{x',z'} \rangle \nonumber\\
&=\frac{1}{|\calC|} \sum_{C \in \calC} \sum_{(x,z)\neq (0^{m+\secpar}, 0^{m+\secpar})}  \big\||\tilde{\phi}_{x,z}\rangle \big\|^2 \nonumber\\
&+ \frac{1}{|\calC|} \sum_{C \in \calC} \sum_{(x,z) \neq (x',z') \neq (0^{m+\secpar}, 0^{m+\secpar})} \textnormal{Tr}\Big[(I \otimes \ketbra{0^{\secpar}}) C^{\dagger}X^{x'}Z^{z'} C (\ketbra{\psi}\otimes \ketbra{0^{\secpar}})C^{\dagger}X^xZ^z C \Big]  \nonumber\\
&= \sum_{(x,z)\neq (0^{m+\secpar}, 0^{m+\secpar})} \big\||\tilde{\phi}_{x,z}\rangle \big\|^2 \frac{1}{|\calC|}\sum_{\substack{C \in \calC:\\ \pi_C^X(x,z) \neq 0^{\secpar}}} 1 \nonumber \\
&+  \sum_{(x,z) \neq (x',z') \neq (0^{m+\secpar}, 0^{m+\secpar})} \textnormal{Tr}\Big[(I \otimes \ketbra{0^{\secpar}}) \frac{1}{|\calC|}\sum_{C \in \calC} C^{\dagger}X^{x'}Z^{z'} C (\ketbra{\psi}\otimes \ketbra{0^{\secpar}})C^{\dagger}X^xZ^z C \Big] \nonumber \\
&=\sum_{(x,z)\neq (0^{m+\secpar}, 0^{m+\secpar})} \big\||\tilde{\phi}_{x,z}\rangle \big\|^2 \frac{1}{|\calC|}\sum_{\substack{C \in \calC:\\ \pi_C^X(x,z) \neq 0^{\secpar}}} 1  \label{eq:24}\\
&= \sum_{(x,z)\neq (0^{m+\secpar}, 0^{m+\secpar})} \big\||\tilde{\phi}_{x,z}\rangle \big\|^2 \cdot 2^{-\secpar} \,, \label{eq:25}
\end{align}
where \eqref{eq:24} follows from the Clifford twirl (\Cref{lem:clifford-twirl}), and \eqref{eq:25} follows from the fact that for any $(x,z)\neq (0^{m+\secpar}, 0^{m+\secpar})$, the fraction of $C \in \calC$ such that $\pi_C^X(x,z) = 0^{\secpar}$ is exactly $2^{-\secpar}$.
We claim that $$\sum_{(x,z)} \big\||\tilde{\phi}_{x,z}\rangle \big\|^2 = 1\,.$$ Assuming this is the case, we have 
\begin{equation}
    \Big\|I_{\calA \calB_1 \calR} \otimes \ketbra{0^{\secpar}}_{\calB_2}\,\,\abs{\calC}^{-1/2} \sum_{C \in \calC} \ket*{C}_{\calA} \otimes \sum_{(x,z)\neq (0^{m+\secpar}, 0^{m+\secpar}) } C^{\dagger} X^{x}Z^{z} C(\ket*{\psi_b} \otimes \ket*{0^\secpar})_{\calB} \otimes  |\tilde{\phi}_{x,z}\rangle \Big\|^2 \leq 2^{-\secpar}\,,
\end{equation}
as desired. We now prove the claim. The calculation is similar to the we just performed. We have
\begin{align}
    1 &= \Big\| \abs{\calC}^{-1/2} \sum_{C \in \calC} \ket*{C}_{\calA} \otimes \sum_{(x,z)} C^{\dagger} X^{x}Z^{z} C(\ket*{\psi_b} \otimes \ket*{0^\secpar})_{\calB} \otimes  |\tilde{\phi}_{x,z}\rangle \Big\|^2 \nonumber \\
    &= \frac{1}{|\calC|}\sum_{C \in \calC} \sum_{x,z} \big\||\tilde{\phi}_{x,z}\rangle \big\|^2 \nonumber\\
    &+ \frac{1}{|\calC|} \sum_{C \in \calC} \sum_{(x,z) \neq (x',z')} \bra{\psi}\bra{0^{\secpar}}  C^{\dagger}X^xZ^z X^{x'}Z^{z'} C^{\dagger}\ket{\psi}\ket{0^{\secpar}} \cdot \langle \tilde{\phi}_{x,z} | \tilde{\phi}_{x',z'} \rangle \nonumber\\
    &=  \sum_{x,z} \big\||\tilde{\phi}_{x,z}\rangle \big\|^2 \nonumber\\
    &+ \frac{1}{|\calC|} \sum_{C \in \calC} \sum_{(x,z) \neq (x',z')} \textnormal{Tr}\Big[ \frac{1}{|\calC|} \sum_{C \in \calC} C^{\dagger}X^{x'}Z^{z'} C (\ketbra{\psi}\otimes \ketbra{0^{\secpar}})C^{\dagger}X^xZ^z C \Big] \nonumber \\
    &=  \sum_{x,z} \big\||\tilde{\phi}_{x,z}\rangle \big\|^2 \,, \nonumber
\end{align}
where the last line follows again by the Clifford twirl. This concludes the proof of \Cref{lem:5}.
\end{proof}

\begin{lemma} \label{lemma:admissible-oracles-app}
    For any adversary $A$ for Hybrid 3,
$$\Pr[A \textnormal{ wins in Hybrid } 3] = \frac{1}{2}.$$
\end{lemma}

The proof of \Cref{lemma:admissible-oracles-app} is a simple application of the ``admissible oracle lemma'' from \cite{GJMZ23}. In fact it is a very special case of that lemma, where the adversary has unbounded computation and the indistinguishability is perfect. We state this simple case of the admissible oracle lemma before presenting our proof of \Cref{lemma:admissible-oracles-app}.

\begin{definition}[$(W,\Pi)$-distinguishing game, \cite{GJMZ23}]
    Let $(\calA, \calB)$ be two quantum registers. Let $W$ be a binary observable and $\Pi$ be a projector on $(\calA, \calB)$ such that $\Pi$ commutes with $W$. Consider the following distinguishing game:
    \begin{enumerate}
        \item The adversary sends a quantum state on registers $(\calA, \calB)$ to the challenger.
        \item The challenger chooses a random bit $b \leftarrow \{0,1\}$. Next, it measures measures $\{\Pi,I-\Pi\}$; if the measurement rejects, abort and output a random bit $b' \from \{0,1\}$. Otherwise, the challenger applies $W^b$ to $(\calA, \calB)$, and returns $\calB$ to the adversary.
        \item The adversary outputs a guess $b'$.
    \end{enumerate}
    We define the distinguishing advantage of the adversary to be $\abs{\Pr[b'=b] - 1/2}$.
\end{definition}

\begin{lemma}[Admissible oracle lemma --- special case \cite{GJMZ23}]
\label{lemma:admissible-oracles}
    Suppose that every adversary achieves zero advantage in the $(W,\Pi)$-distinguishing game. Let $G$ be an admissible unitary, i.e.,
    \begin{itemize}
        \item $G$ commutes with both $W$ and $\Pi$, and
        \item $G$ acts identically on $I-\Pi$, i.e., $G(I-\Pi) = I - \Pi$.
    \end{itemize}
    Then every adversary achieves zero advantage in the $(W, \Pi)$-distinguishing game, even when given oracle access to $G$.
\end{lemma}

\begin{proof}[Proof of \Cref{lemma:admissible-oracles-app}]
We apply \Cref{lemma:admissible-oracles} with the following choices of $\Pi, W, G$:
    \begin{itemize}
        \item $\Pi$ is the projection on $\calA, \calB$ to all states of the form $\abs{\calC}^{-1/2} \sum_{C \in \calC} \ket*{C} \otimes C(\ket*{\psi} \otimes \ket*{0^\secpar})$, where $\calC$ is the Clifford group and $\ket*{\psi}$ is any state.
        \item For $b \in \{0,1\}$, $\ket*{\tilde{\psi}_b} = \abs{\calC}^{-1/2} \sum_{C \in \calC} \ket*{C} \otimes C(\ket*{\psi_b} \otimes \ket*{0^\secpar})$ are the states to be distinguished.
        \item $W = \ketbra*{\tilde{\psi}_0}{\tilde{\psi}_1} + \ketbra*{\tilde{\psi}_1}{\tilde{\psi}_0} + (I-\ketbra*{\tilde{\psi}_0}{\tilde{\psi}_0} - \ketbra*{\tilde{\psi}_1}{\tilde{\psi}_1})$ is the operator that swaps $\ket*{\tilde{\psi}_0}$ and $\ket*{\tilde{\psi}_1}$ and otherwise acts as the identity on the subspace orthogonal to the span of $\ket*{\tilde{\psi}_0}$ and $\ket*{\tilde{\psi}_1}$. 
        \item $G$ is the unitary that acts as $\Eval$ on the range of $\Pi$, and acts as the identity on the range of $I - \Pi$.
    \end{itemize}
    It is immediate from the fact that the Clifford group is a unitary 1-design that the $(W, \Pi)$-distinguishing game has perfect security. It is also easy to see that $G$ is an admissible oracle for $(W, \Pi)$. Therefore, \Cref{lemma:admissible-oracles} implies that no adversary can obtain any advantage for distinguishing between the $\calB$ registers of $\ket*{\tilde{\psi}_0}$ and $\ket*{\tilde{\psi}_1}$, even given oracle access to $G$.
\end{proof}
This completes the proof of \Cref{theorem:construction}.
\end{proof}

\section{Unclonable Encryption}
\label{sec:unclonable-primitives}
In this section we introduce a variant of unclonable encryption (UE) that we call \emph{coupled unclonable encryption} (cUE). Coupled unclonable encryption is a weaker primitive than UE, in the sense that any secure UE scheme can be used to build a secure cUE scheme. It closely resembles UE, the main difference being that in cUE there are \emph{two} encryption keys that decrypt \emph{two} messages. The main result of this section is that, unlike UE --- which we do not know how to construct from standard assumptions --- we can build cUE from one-way functions. The main technical ideas behind our construction are presented in \Cref{subsec:unc-rand}, and the cUE construction and proof of security are given in \Cref{subsec:cue}.

Beyond being interesting in its own right, we will show in \Cref{sec:cp} that, in conjunction with $\qsiO$, cUE is already sufficient to build copy protection for certain interesting classes of functions. In order to obtain these applications, we will need an additional feature of UE or cUE that we call \emph{key testing}. This feature is described in \Cref{subsec:key-testing}, where we also show that key testing can be generically added to any UE or cUE scheme using $\qsiO$ and injective one-way functions.

For an outline of the ideas and techniques used in this section, see the technical overview (\Cref{subsec:techo}).

\subsection{Unclonable randomness} \label{subsec:unc-rand}
We find it is easier to reason about a slightly weaker primitive than cUE, which we call ``unclonable randomness.'' Essentially, unclonable randomness is cUE but for random messages that the adversary does not choose: it allows one to encrypt \emph{random} strings $r,s$ under a secret key. The security guarantee says that it is not possible to split the encryption into two states which can both be used (together with the secret key) to learn \emph{any} information about $r$ and $s$.

Since we are able to build unclonable randomness unconditionally, and since it is just a building block for our cUE construction, we only formally define it for our particular construction (rather than as an abstract primitive). The security game for our unclonable randomness construction is given in \Cref{game:rand}, and security is proven in \Cref{theorem:rand-indep}. This result can be viewed as a decision version of the main result of \cite{TFKW13}.

\begin{theorem} \label{theorem:rand-indep}
    For any computationally unbounded adversary $\Adv$, and any $n, \secpar \in \mathbb{N}$,
    \[
        \Pr[\emph{\RandExpt}_{\Adv}(n,\secpar)=1] \le \frac{1}{2} + \poly(n) \cdot 2^{-\Omega(\lambda)}  \,,
    \]
    where $\emph{\RandExpt}_{\Adv}(n,\secpar)$ is described in \Cref{game:rand}.
\end{theorem}

In particular, when $n= \poly(\lambda)$, the advantage is negligible in $\lambda$.

\begin{figure}[H]
\pcb{
    \textbf{Challenger} \< \< \textbf{Adversary} \\[][\hline]
    \< \< \\[-0.5\baselineskip]
    x, \theta \from \{0,1\}^{10n+\secpar},\ U, V \from \{0,1\}^{n \times (10n+\secpar)} \<\< \\
    r^0, s^0 \from \{0,1\}^n \\
    r^1 := U x,\ s^1 := V x \<\< \\
    a, b \from \{0,1\} \<\< \\
    \< \sendmessageright*{\ket*{x^\theta}, r^a, s^b} \< \\
    \<\< (A, B, \rho_{\calA,\calB}) \from \Adv(\ket*{x^\theta}, r^a, s^b) \\
    \< \sendmessageleft*{A, B, \rho_{\calA,\calB}} \< \\
    (a',b') \from A^{\theta,U} \otimes B^{\theta,V} \,\rho_{\calA,\calB} \<\< \\
    \textnormal{Output 1 if $a'=a$ and $b'=b$; otherwise output 0} \<\<
}
\caption{$\RandExpt_{\Adv}(n,\secpar)$. The challenger first generates random strings $x, \theta \from \{0,1\}^{10n+\secpar}, r^0, s^0 \from \{0,1\}^n$ and random matrices $U, V \from \{0,1\}^{n \times (10n+\secpar)}$. It then computes $r^1$ and $s^1$ as $Ux$ and $Vx$ respectively. The challenger samples random bits $a$ and $b$, and sends the state $\ket*{x^\theta}$ along with $r^a$ and $s^b$ to the adversary. The adversary then computes a quantum state $\rho_{\calA,\calB}$ and circuit descriptions $A$ and $B$, and sends $(A, B, \rho_{\calA,\calB})$ back to the challenger. The challenger measures $A^{\theta,U}$ and $B^{\theta,V}$ on $\rho_{\calA,\calB}$, obtaining outcomes $a'$ and $b'$. The adversary wins if $a'=a$ and $b'=b$.} \label{game:rand}
\end{figure}

We prove \Cref{theorem:rand-indep} by reduction from a search version of the same game, defined in \Cref{game:search}. That this search version is secure follows straightforwardly from the results of \cite{TFKW13}, and is proven in \Cref{corollary:tfkw-mod}.

\begin{figure}[H]
\pcb{
    \textbf{Challenger} \< \< \textbf{Adversary} \\[][\hline]
    \< \< \\[-0.5\baselineskip]
    x, \theta \from \{0,1\}^{10n+\secpar},\ U, V \from \{0,1\}^{n \times (10n+\secpar)} \<\< \\
    \< \sendmessageright*{\ket*{x^\theta}, U x, V x} \< \\
    \<\< (\tilde{A}, \tilde{B}, \rho_{\calA,\calB}) \from \Adv(\ket*{x^\theta}, U x, V x) \\
    \< \sendmessageleft*{\tilde{A}, \tilde{B}, \rho_{\calA,\calB}} \< \\
    (x_A , x_B) \from \tilde{A}^{\theta,U} \otimes \tilde{B}^{\theta,V}\,\rho_{\calA,\calB}\<\< \\
    \textnormal{Output 1 if $x_A=x_B=x$; otherwise output 0} \<\<
}
\caption{$\SearchExpt_{\Adv}(n,\secpar)$. The challenger generates random strings $x, \theta \from \{0,1\}^{10n+\secpar}$ and matrices $U, V \from \{0,1\}^{n \times (10n+\secpar)}$ and sends $\ket*{x^\theta}, Ux, Vx$ to the adversary. The adversary responds with quantum circuits $\tilde{A}, \tilde{B}$ acting on a state $\rho_{\calA,\calB}$. The challenger measures $\tilde{A}^{\theta,U}$ and $\tilde{B}^{\theta,V}$ on $\rho_{\calA,\calB}$, obtaining outputs $x_A$ and $x_B$. The adversary wins if $x_A=x_B=x$.} \label{game:search}
\end{figure}

The $n=0$ case is exactly the monogamy-of-entanglement game considered in \cite{TFKW13}. Observe that when $n=0$, $U$ and $V$ are empty and the first message is simply $\ket*{x^\theta}$. When $n > 0$, the adversary is given some extra information about $x$.

\begin{theorem}[Theorem 3 in \cite{TFKW13}] \label{theorem:tfkw}
    For $\secpar \in \mathbb{N}$ and any computationally unbounded adversary $\Adv$,
    \[
        \Pr[\emph{\SearchExpt}_{\Adv}(0,\secpar)=1] \le \left(\frac{1}{2} + \frac{1}{2\sqrt{2}}\right)^\secpar.
    \]
\end{theorem}

In the following corollary of \Cref{theorem:tfkw}, we show that the general $\SearchExpt$ reduces to the $n=0$ case of $\SearchExpt$.

\begin{corollary} \label{corollary:tfkw-mod}
    For $n, \secpar \in \mathbb{N}$ and any computationally unbounded adversary $\Adv$,
    \[
        \Pr[\emph{\SearchExpt}_{\Adv}(n,\secpar)=1] \le \left(\frac{1}{2} + \frac{1}{2\sqrt{2}}\right)^\secpar.
    \]
\end{corollary}
\begin{proof}
    Suppose that $\Adv$ obtains advantage $\varepsilon$ in $\SearchExpt_{\Adv}(n,\secpar)$. We design a reduction $\Red[\Adv,n,\secpar]$ that uses $\Adv$ to play $\SearchExpt(0,10n+\secpar)$ by simply \emph{guessing} the values $Ux$ and $Vx$. Formally, $\Red$ is defined by the following behavior in $\SearchExpt(0,10n+\secpar)$.
    
    \indent $\SearchExpt_{\Red[\Adv,n,\secpar]}(0,10n+\secpar)$:
    \begin{enumerate}
        \item The challenger samples $x, \theta \from \{0,1\}^{10n+\secpar}$ and sends $\ket{x^\theta}$ to the reduction.
        \item The reduction samples $r, s \from \{0,1\}^n$ and sends $(\ket*{x^\theta}, r, s)$ to $\Adv$.
        \item $\Adv$ outputs a state $\rho_{\calA,\calB}$ and descriptions of $2^{10n+\secpar}$-outcome measurement families $\tilde{A}, \tilde{B}$.
        \item The reduction samples $U, V \from \{0,1\}^{n \times (10n + \secpar)}$ and returns $(\tilde{A}^{\cdot,U}, \tilde{B}^{\cdot,V}, \rho_{\calA,\calB})$ to the challenger.
        \item The challenger obtains $x_A \from \tilde{A}^{\theta,U}(\rho_{\calA})$ and $x_B \from \tilde{B}^{\theta,V}(\rho_{\calB})$. The adversary wins if $x_A=x_B=x$.
    \end{enumerate}
    The probability that $\Red$ samples $U, V, r, s$ such that $U x = r$ and $V x = s$ is $2^{-2n}$, and conditioned on this event the view of the adversary in $\SearchExpt(n, \secpar)$ is exactly reproduced. Therefore, $\Red$ has advantage $\varepsilon / 2^{2n}$ in $\SearchExpt(0,10n+\secpar)$. By \Cref{theorem:tfkw}, we have
    \begin{align*}
        \varepsilon &\le 2^{2n} \cdot \left(\frac{1}{2} + \frac{1}{2\sqrt{2}}\right)^{10n+\secpar} \\
        &= \left[4 \cdot \left(\frac{1}{2} + \frac{1}{2\sqrt{2}}\right)^{10}\right]^n \cdot \left(\frac{1}{2} + \frac{1}{2\sqrt{2}}\right)^{\secpar} \\
        &\le \left(\frac{1}{2} + \frac{1}{2\sqrt{2}}\right)^{\secpar}. \qedhere
    \end{align*}
\end{proof}

In order to reduce the security of $\RandExpt$ to that of $\SearchExpt$, and prove \Cref{theorem:rand-indep}, we require two lemmas.

\begin{lemma}
\label{lemma:simult}
	Let $\{\ket{\psi_z}\}_{z \in \calZ}$ be a family of states and $\{P_z, Q_z\}_{z \in \calZ}$ be a family of operators. Suppose that $0 \le P_z, Q_z \le 1$ for all $z \in \calZ$ and
	$$\E_{z \from \calZ} \bra{\psi_z} \left(\frac{1+P_z}{2}\right) \otimes \left(\frac{1+Q_z}{2}\right) \ket{\psi_z} \ge \frac{1}{2} + \varepsilon.$$
	Then
	$$\E_{z \from \calZ} \bra{\psi_z} P_z \otimes Q_z \ket{\psi_z} \ge \varepsilon^3.$$
\end{lemma}
\begin{proof}
	Let $\{\ket{\phi_z^i}\}_i$ and $\{\ket{\tau_z^j}\}_j$ be eigenbases for $P_z$ and $Q_z$, respectively. Then we can write
	$$\ket{\psi_z} = \sum_{i,j} \alpha_z^{i,j} \ket{\phi_z^i} \otimes \ket{\tau_z^j},$$
	and letting $\alpha_z$ denote the distribution over $(i,j)$ with probabilities $\abs{\alpha_z^{i,j}}^2$ we have
	$$\E_{\substack{z \from \calZ \\ (i,j) \from \alpha_z}} \left( \frac{1}{2} + \frac{1}{2} \bra{\phi_z^i} P_z \ket{\phi_z^i}\right) \left( \frac{1}{2} + \frac{1}{2} \bra{\tau_z^j} Q_z \ket{\tau_z^j}\right) \ge \frac{1}{2} + \varepsilon.$$
	By an averaging argument, it follows that
	$$\Pr_{\substack{z \from \calZ \\ (i,j) \from \alpha_z}}\left[\bra{\phi_z^i} P_z \ket{\phi_z^i} \ge \varepsilon \textnormal{ and } \bra{\tau_z^j} Q_z \ket{\tau_z^j} \ge \varepsilon\right] \ge \varepsilon.$$
	Therefore
	\begin{align*}
		\E_{z \from \calZ} \bra{\psi_z} P_z \otimes Q_z \ket{\psi_z} &= \E_{\substack{z \from \calZ \\ (i,j) \from \alpha_z}} \bra{\phi_z^i} P_z \ket{\phi_z^i} \bra{\tau_z^j} Q_z \ket{\tau_z^j} \\
		&\ge \varepsilon^3. \qedhere
	\end{align*}
\end{proof}

A central component in our proof of \Cref{theorem:rand-indep} is the quantum Goldreich-Levin reduction of \cite{BV97,AC02}. We recall that algorithm here. Let $\{A^u\}_{u \in \{0,1\}^n}$ be a collection of binary-outcome measurements and let $\ket{\psi}$ be a state.

\indent $\GL(\{A^u\}_{u \in \{0,1\}^n})$:
    \begin{enumerate}
        \item Prepare the state
        \[
            2^{-n/2} \sum_{u \in \{0,1\}^n} \ket{u} \otimes \ket{\psi}.
        \]
        \item Apply $\sum_{u \in \{0,1\}^n} \ketbra{u} \otimes A_{ph}^u$, where $A_{ph}^u$ is the phase oracle for $A^u$ --- i.e., $A_{ph}^u$ applies a phase of $(-1)$ to the subspace where $A^u = 1$ and acts as identity on the subspace where $A^u$.
        \item Measure the $\ket{u}$ register in the Hadamard basis, and output the result.
    \end{enumerate}

\begin{lemma}[Simultaneous quantum Goldreich-Levin computation]
\label{lemma:bv}

Let $\{A^u\}_{u \in \{0,1\}^n}$ and $\{B^v\}_{v \in \{0,1\}^n}$ be collections of binary-outcome measurements that act on disjoint registers $\calA$ and $\calB$. Let $\ket{\psi}$ be a state on $\calA, \calB$. Then the probability that $\GL(\{A^u\}_{u \in \{0,1\}^n}) \otimes \GL(\{B^v\}_{v \in \{0,1\}^n}) \ket{\psi}$ returns $(x,x)$ is
\[
    \Pr\Big[(x,x) \from \GL(\{A^u\}_{u \in \{0,1\}^n}) \otimes \GL(\{B^v\}_{v \in \{0,1\}^n}) \ket{\psi}\Big] = \norm{(2 \E_u \Pi_A^{x,u} - I) \otimes (2 \E_v \Pi_B^{x,v} - I) \ket{\psi}}^2,
\]
where $\Pi_A^{x,u}$ and $\Pi_B^{x,u}$ are the projections onto the subspaces where $A$ and $B$ output $u \cdot x$ and $v \cdot x$, respectively.
\end{lemma}
\begin{proof}
    The first step of $\GL(\{A^u\}_{u \in \{0,1\}^n}) \otimes \GL(\{B^v\}_{v \in \{0,1\}^n})$ is to prepare the state
    $$2^{-n} \sum_{u, v \in \{0,1\}^n} \ket{u,v} \otimes \ket{\psi}.$$
    We then apply our controlled phase oracles to get the state
    \begin{align*}
		& 2^{-n} \sum_{u, v \in \{0,1\}^n} \ket{u,v} \otimes (A^u \otimes B^v) \ket{\psi} \\
		&= 2^{-n} \sum_{u, v \in \{0,1\}^n} \ket{u,v} \otimes (-1)^{(u+v) \cdot x} \cdot (2 \cdot \Pi_A^{x,u} - 1) \otimes (2 \cdot \Pi_B^{x,v} - 1) \ket{\psi},
    \end{align*}
    where we have used the fact that
    \[
        A_{ph}^u \ket{\psi} = (-1)^{u \cdot x} \cdot \Pi_A^{x,u} \ket{\psi} + (-1)^{1-u \cdot x} \cdot (1-\Pi_A^{x,u}) \ket{\psi} = (-1)^{u \cdot x} \cdot (2 \cdot \Pi_A^{x,u} - 1) \ket{\psi}
    \]
    (and similarly for $B_{ph}^v$).
    
    Next we apply Hadamard gates to the $\ket{u,v}$ part, project onto $\ketbra{x,x}$, and take the norm squared to find the probability that both $\GL(\{A^u\}_{u \in \{0,1\}^n})$ and $\GL(\{B^v\}_{v \in \{0,1\}^n})$ output $x$. That quantity is
    \begin{align*}
        & \norm{2^{-2n} \sum_{u, v \in \{0,1\}^n} \ket{x,x} \otimes (2 \cdot \Pi_A^{x,u} - 1) \otimes (2 \cdot \Pi_B^{x,v} - 1) \ket{\psi}}^2 \\
        &= \norm{\E_{u, v \in \{0,1\}^n} (2 \cdot \Pi_A^{x,u} - 1) \otimes (2 \cdot \Pi_B^{x,v} - 1) \ket{\psi}}^2. \qedhere
    \end{align*}
\end{proof}

\begin{proof}[Proof of \Cref{theorem:rand-indep}]
    As anticipated, we reduce the security of $\RandExpt$ to that of $\SearchExpt$ (\Cref{game:search}), which we established in \Cref{corollary:tfkw-mod}. The first step of the proof is to rewrite $\SearchExpt$ in an equivalent form.
    
    Observe that the challenger in $\SearchExpt$ could sample the vectors corresponding to $Ux, Vx$ \emph{before} actually deciding on the matrices $U, V$, and the security game would be identical. That is, the challenger will sample random vectors $r, s \from \{0,1\}^n$ and send $\ket*{x^\theta}, r, s$ to the adversary in the first step. Later, in order to run $\tilde{A}, \tilde{B}$, it will just sample random matrices $U,V$ conditioned on $Ux = r$ and $Vx = s$.

    Before we give a formal description of this equivalent game, we define distributions that will be useful throughout the proof. For $i \in [n]$, $x \in \{0,1\}^{10n+\secpar}$, and $r \in \{0,1\}^n$, let $\calD_i(x,r)$ be the distribution over matrices $U \in \{0,1\}^{n \times (10n+\secpar)}$ where row $j \in [n]$ is sampled as
    $$U_j = \begin{cases}
        u_j \from \{u \in \{0,1\}^{10n+\secpar} \mid u \cdot x = r_j\}, & j \le i \\
        u_j \from \{0,1\}^{10n+\secpar}, & j > i.
    \end{cases}$$
    That is, $U \from \calD_i(x,r)$ is a random matrix conditioned on the first $i$ values of $Ux$ being equal to the first $i$ values of $r$.
    Now our equivalent formulation of $\SearchExpt$ is as follows. Without loss of generality we assume that the state $\rho_{\calA,\calB}$ is a pure state $\ket{\psi_{x,\theta,r,s}}$.
    
    \indent $\SearchExpt_{\Adv}(n,\secpar)$:
    \begin{enumerate}
        \item The challenger samples $x, \theta \from \{0,1\}^{10n+\secpar}$ and $r,s \from \{0,1\}^n$, then sends $\ket{x^\theta}, r, s$ to $\Adv$.
        \item $\Adv$ outputs a state $\ket{\psi_{x,\theta,r,s}}$ and descriptions of $2^{10n+\secpar}$-outcome measurement families $\tilde{A}, \tilde{B}$.
        \item The challenger samples $U \from D_{n}(x,r), V \from D_{n}(x,s)$ and measures $\tilde{A}^{\theta, U}$ and $\tilde{B}^{\theta, V}$ on $\ket{\psi_{x,\theta,r,s}}$, obtaining outcomes $x_A$ and $x_B$. The adversary wins if $x_A=x_B=x$.
    \end{enumerate}
    We similarly rewrite $\RandExpt$ in an equivalent form where the challenger decides on the matrices at the end of the game.
    
    \indent $\RandExpt_{\Adv}(n,\secpar)$:
    \begin{enumerate}
        \item The challenger samples $x, \theta \from \{0,1\}^{10n+\secpar}$ and $r,s \from \{0,1\}^n$, then sends $\ket{x^\theta}, r, s$ to $\Adv$.
        \item $\Adv$ outputs a state $\ket{\psi_{x,\theta,r,s}}$ and descriptions of binary-outcome measurement families $A, B$.
        \item The challenger samples $a, b \from \{0,1\}$, then samples $U \from D_{a \cdot n}(x,r), V \from D_{b \cdot n}(x,s)$ and measures $A^{\theta, U}$ and $B^{\theta, V}$ on $\ket{\psi_{x,\theta,r,s}}$, obtaining outcomes $a'$ and $b'$. The adversary wins if $a'=a$ and $b'=b$.
    \end{enumerate}
    Note that the first message is the same in both $\SearchExpt$ and $\RandExpt$, and the challenge bits $a,b$ in $\RandExpt$ are sampled independently from the first message. Given an adversary $\Adv$ for $\RandExpt_{\Adv}(n,\secpar)$, we define an adversary $\Red[\Adv]$ for $\SearchExpt$. The latter uses the \cite{BV97,AC02} reduction where the adversaries guess random bits of $r, s$.
    
    \indent $\SearchExpt_{\Red[\Adv]}(n,\secpar)$:
    \begin{enumerate}
        \item The challenger samples $x, \theta \from \{0,1\}^{10n+\secpar}$ and $r,s \from \{0,1\}^n$, then sends $\ket{x^\theta}, r, s$ to $\Red$, which forwards everything to $\Adv$.
        \item $\Adv$ outputs a state $\ket{\psi_{x,\theta,r,s}}$ and descriptions of binary-outcome measurement families $A, B$.
        \item $\Red$ samples random $i, j \from [n]$ and $\tilde{U}, \tilde{V} \from \{0,1\}^{n \times (10n+\secpar)}$. Let $\tilde{A}^{\theta,U} := \GL(\{A^{\theta,[U_{<i} || u || \tilde{U}_{>i}]}\}_{u \in \{0,1\}^n})$ and $\tilde{B}^{\theta,V} := \GL(\{B^{\theta,[V_{<i} || v || \tilde{V}_{>i}]}\}_{v \in \{0,1\}^n})$. $\Red$ sends $\tilde{A}, \tilde{B}, \ket{\psi_{x,\theta,r,s}}$ to the challenger.
        \item The challenger samples $U \from D_{n}(x,r), V \from D_{n}(x,s)$ and measures $\tilde{A}^{\theta, U}$ and $\tilde{B}^{\theta, V}$ on $\ket{\psi_{x,\theta,r,s}}$, obtaining outcomes $x_A$ and $x_B$. The reduction wins if $x_A=x_B=x$.
    \end{enumerate}
    Suppose that
    \[
        \Pr[\RandExpt_{\Adv}(n,\secpar)=1] \ge \frac{1}{2} + \varepsilon.
    \]
    In $\RandExpt$, we denote the projections onto outcomes $a$ and $b$ by $A_a$ and $B_b$, respectively. Then,
    \begin{align*}
        & \frac{1}{2} + \varepsilon \\
        &\le \E_{\substack{x, \theta \from \{0,1\}^{10n+\secpar} \\ r,s \from \{0,1\}^n}} \bra{\psi_{x,\theta,r,s}} \left(\frac{\E_{U \from D_0(x,r)}[A_0^{\theta,U}]+\E_{U \from D_n(x,r)}[A_1^{\theta,U}]}{2}\right) \\
        & \hspace{4cm} \otimes \left(\frac{\E_{V \from D_0(x,s)}[B_0^{\theta,V}]+\E_{V \from D_n(x,s)}[B_1^{\theta,V}]}{2}\right) \ket{\psi_{x,\theta,r,s}} \\
        &= \E_{\substack{x, \theta \from \{0,1\}^{10n+\secpar} \\ r,s \from \{0,1\}^n}} \bra{\psi_{x,\theta,r,s}} \left(\frac{1 + \E_{U \from D_0(x,r)}[A_0^{\theta,U}]-\E_{U \from D_n(x,r)}[A_0^{\theta,U}]}{2}\right) \\
        & \hspace{4cm} \otimes \left(\frac{1 + \E_{V \from D_0(x,s)}[B_0^{\theta,V}]-\E_{V \from D_n(x,s)}[B_0^{\theta,V}]}{2}\right) \ket{\psi_{x,\theta,r,s}} \\
        &\le \E_{\substack{x, \theta \from \{0,1\}^{10n+\secpar} \\ r,s \from \{0,1\}^n}} \bra{\psi_{x,\theta,r,s}} \left(\frac{1 + \abs{\E_{U \from D_0(x,r)}[A_0^{\theta,U}]-\E_{U \from D_n(x,r)}[A_0^{\theta,U}]}}{2}\right) \\
        & \hspace{4cm} \otimes \left(\frac{1 + \abs{\E_{V \from D_0(x,s)}[B_0^{\theta,V}]-\E_{V \from D_n(x,s)}[B_0^{\theta,V}]}}{2}\right) \ket{\psi_{x,\theta,r,s}}.
    \end{align*}
    By \Cref{lemma:simult}, we have
    \begin{align*}
        & \varepsilon^3 \le \\
        & \E_{\substack{x, \theta \from \{0,1\}^{10n+\secpar} \\ r,s \from \{0,1\}^n}} \bra{\psi_{x,\theta,r,s}} \abs{\E_{U \from D_0(x,r)}[A_0^{\theta,U}]-\E_{U \from D_n(x,r)}[A_0^{\theta,U}]} \otimes \abs{\E_{V \from D_0(x,s)}[B_0^{\theta,V}]-\E_{V \from D_n(x,s)}[B_0^{\theta,V}]} \ket{\psi_{x,\theta,r,s}}.
    \end{align*}
    Next we use a sort of ``hybrid argument'' to relate the \emph{operators} in the two games:
    \begin{align*}
        & \abs{\E_{U \from D_0(x,r)}[A_0^{\theta,U}]-\E_{U \from D_n(x,r)}[A_0^{\theta,U}]} \\
        &= \abs{\sum_{i=1}^n \left(\E_{U \from D_i(x,r)}[A_0^{\theta,U}]-\E_{U \from D_{i-1}(x,r)}[A_0^{\theta,U}]\right)} \\
        &= \abs{\sum_{i=1}^n \E_{\substack{U \from D_n(x,r) \\ \tilde{U} \from D_0(x,r)}} \left(\E_{u : u \cdot x = r_i}[A_0^{\theta,[U_{<i} || u || \tilde{U}_{>i}]}]-\E_{u}[A_0^{\theta,[U_{<i} || u || \tilde{U}_{>i}]}]\right)} \\
        &= \frac{1}{2} \abs{\sum_{i=1}^n \E_{\substack{U \from D_n(x,r) \\ \tilde{U} \from D_0(x,r)}} \left(\E_{u : u \cdot x = r_i}[A_0^{\theta,[U_{<i} || u || \tilde{U}_{>i}]}]-\E_{u : u \cdot x \ne r_i}[A_0^{\theta,[U_{<i} || u || \tilde{U}_{>i}]}]\right)} \\
        &\le \frac{1}{2} \sum_{i=1}^n \E_{\substack{U \from D_n(x,r) \\ \tilde{U} \from D_0(x,r)}} \abs{\E_{u : u \cdot x = r_i}[A_0^{\theta,[U_{<i} || u || \tilde{U}_{>i}]}]-\E_{u : u \cdot x \ne r_i}[A_0^{\theta,[U_{<i} || u || \tilde{U}_{>i}]}]} \\
        &= \frac{n}{2} \E_{\substack{i \from [n] \\ U \from D_n(x,r) \\ \tilde{U} \from D_0(x,r)}} \abs{2 \E_u A_{u \cdot x + r_i}^{\theta,[U_{<i} || u || \tilde{U}_{>i}]} - 1}.
    \end{align*}
    The above holds identically for the $B$ part. We are now ready to bound the probability that our reduction wins $\SearchExpt$. We begin by applying \Cref{lemma:bv}:
    \begin{align*}
        & \Pr[\SearchExpt_{\Red[\Adv]}(n,\secpar)=1] \\
        &= \E_{\substack{x, \theta \from \{0,1\}^{10n+\secpar} \\ r,s \from \{0,1\}^n}} 
 \E_{\substack{i,j \from [n] \\ U \from D_n(x,r), \tilde{U} \from D_0(x,r) \\ V \from D_n(x,s), \tilde{V} \from D_0(x,s)}} \norm{(2 \E_u A_{u \cdot x}^{\theta,[U_{<i} || u || \tilde{U}_{>i}]} -1) \otimes (2 \E_v B_{v \cdot x}^{\theta,[V_{<j} || v || \tilde{V}_{>j}]} - 1) \ket{\psi_{x,\theta,r,s}}}^2 \\
        &= \E_{\substack{x, \theta \from \{0,1\}^{10n+\secpar} \\ r,s \from \{0,1\}^n}} 
 \E_{\substack{i,j \from [n] \\ U \from D_n(x,r), \tilde{U} \from D_0(x,r) \\ V \from D_n(x,s), \tilde{V} \from D_0(x,s)}} \bra{\psi_{x,\theta,r,s}} \abs{2 \E_u A_{u \cdot x}^{\theta,[U_{<i} || u || \tilde{U}_{>i}]} -1}^2 \otimes \abs{2 \E_v B_{v \cdot x}^{\theta,[V_{<j} || v || \tilde{V}_{>j}]} - 1}^2 \ket{\psi_{x,\theta,r,s}}\\
        &= \E_{\substack{x, \theta \from \{0,1\}^{10n+\secpar} \\ r,s \from \{0,1\}^n}} 
 \E_{\substack{i,j \from [n] \\ U \from D_n(x,r), \tilde{U} \from D_0(x,r) \\ V \from D_n(x,s), \tilde{V} \from D_0(x,s)}} \bra{\psi_{x,\theta,r,s}} \abs{2 \E_u A_{u \cdot x+r_i}^{\theta,[U_{<i} || u || \tilde{U}_{>i}]} -1}^2 \otimes \abs{2 \E_v B_{v \cdot x+s_j}^{\theta,[V_{<j} || v || \tilde{V}_{>j}]} - 1}^2 \ket{\psi_{x,\theta,r,s}}
        \end{align*}
where the last line is because $\abs{2 \E_u A_{u \cdot x}^{\theta,[U_{<i} || u || \tilde{U}_{>i}]} -1} = \abs{2 \E_u A_{u \cdot x+1}^{\theta,[U_{<i} || u || \tilde{U}_{>i}]} -1}$, and similarly for the $B$ term. Finally, we have
\begin{align*}
    &\E_{\substack{x, \theta \from \{0,1\}^{10n+\secpar} \\ r,s \from \{0,1\}^n}} \E_{\substack{i,j \from [n] \\ U \from D_n(x,r), \tilde{U} \from D_0(x,r) \\ V \from D_n(x,s), \tilde{V} \from D_0(x,s)}} \bra{\psi_{x,\theta,r,s}} \abs{2 \E_u A_{u \cdot x+r_i}^{\theta,[U_{<i} || u || \tilde{U}_{>i}]} -1}^2 \otimes \abs{2 \E_v B_{v \cdot x+s_j}^{\theta,[V_{<j} || v || \tilde{V}_{>j}]} - 1}^2 \ket{\psi_{x,\theta,r,s}} \\
    &=\E_{\substack{x, \theta \from \{0,1\}^{10n+\secpar} \\ r,s \from \{0,1\}^n}} \E_{\substack{i,j \from [n] \\ U \from D_n(x,r), \tilde{U} \from D_0(x,r) \\ V \from D_n(x,s), \tilde{V} \from D_0(x,s)}} \left( \bra{\psi_{x,\theta,r,s}} \abs{2 \E_u A_{u \cdot x+r_i}^{\theta,[U_{<i} || u || \tilde{U}_{>i}]} -1} \otimes \abs{2 \E_v B_{v \cdot x+s_j}^{\theta,[V_{<j} || v || \tilde{V}_{>j}]} - 1}  \ket{\psi_{x,\theta,r,s}} \right)^2  \\
    &\ge \E_{\substack{x, \theta \from \{0,1\}^{10n+\secpar} \\ r,s \from \{0,1\}^n}} \bra{\psi_{x,\theta,r,s}} \left(\E_{\substack{i \from [n] \\ U \from D_n(x,r) \\ \tilde{U} \from D_0(x,r)}} \abs{2 \E_u A_{u \cdot x+r_i}^{\theta,[U_{<i} || u || \tilde{U}_{>i}]} -1}\right) \otimes \left(\E_{\substack{j \from [n] \\ V \from D_n(x,s) \\ \tilde{V} \from D_0(x,s)}} \abs{2 \E_v B_{v \cdot x+s_j}^{\theta,[V_{<j} || v || \tilde{V}_{>j}]} -1}\right) \ket{\psi_{x,\theta,r,s}}^2 \\
    &\ge \frac{4}{n^2} \E_{\substack{x, \theta \from \{0,1\}^{10n+\secpar} \\ r,s \from \{0,1\}^n}} \bra{\psi_{x,\theta,r,s}} \abs{\E_{U \from D_0(x,r)}[A_0^{\theta,U}]-\E_{U \from D_n(x,r)}[A_0^{\theta,U}]} \otimes \abs{\E_{V \from D_0(x,s)}[B_0^{\theta,V}]-\E_{V \from D_n(x,s)}[B_0^{\theta,V}]} \ket{\psi_{x,\theta,r,s}}^2 \\
    &\ge \frac{4 \varepsilon^6}{n^2}.
\end{align*}
where the first inequality is by convexity. By \Cref{corollary:tfkw-mod}, this quantity is at most $\left(\frac{1}{2} + \frac{1}{2\sqrt{2}}\right)^\secpar$. Therefore, $\varepsilon \le n^{1/3} \cdot 2^{-\Omega(\lambda)}$.
\end{proof}

\subsection{Coupled unclonable encryption}
\label{subsec:cue}
In this subsection, we introduce \emph{coupled unclonable encryption} (cUE). It is similar to UE, except that it involves the simultaneous encryption of two messages $m_A$ and $m_B$ under two secret keys $\sk_A$ and $\sk_B$. Informally, security for cUE says that when a pirate processes the ciphertext into two parts, one given to Alice and the other to Bob, then after receiving $\sk_A$ and $\sk_B$ it is not possible for both of Alice and Bob to simultaneously recover any information about their respective messages $m_A$ and $m_B$. While cUE is weaker than UE, we are able to make use of it in \Cref{sec:cp} as a central primitive in our proofs that $\qsiO$ copy-protects puncturable programs.

Before introducing cUE in \Cref{def:cue}, we recall the definition of UE in \Cref{def:ue}.

\begin{figure}
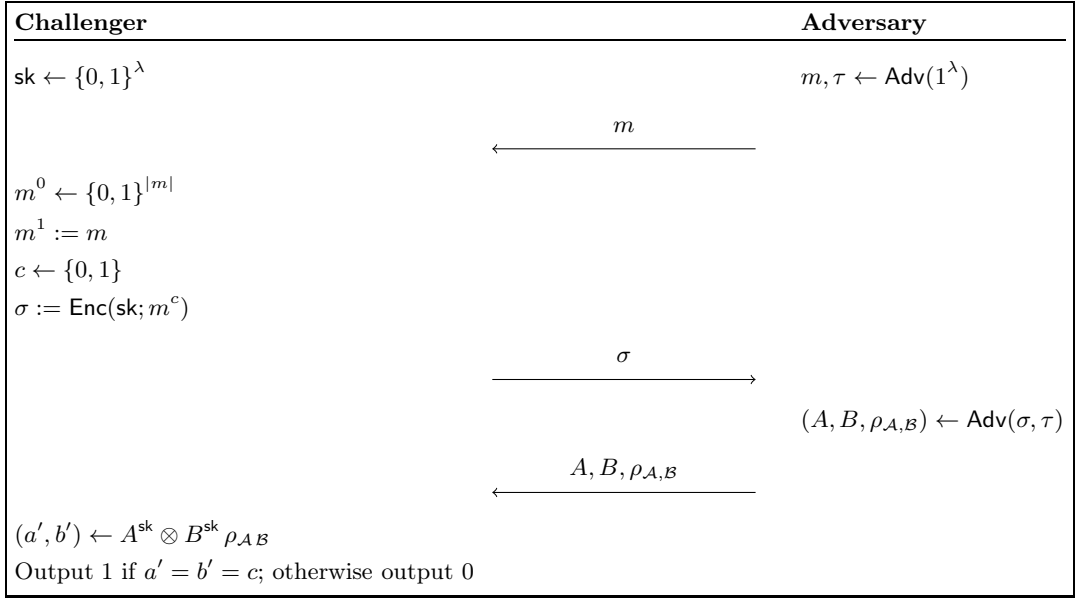

\pcb{
    \textbf{Challenger} \< \< \textbf{Adversary} \\[][\hline]
    \< \< \\[-0.5\baselineskip]
    \sk \from \{0,1\}^\secpar \<\< m, \tau \from \Adv(\secparam) \\
    \< \sendmessageleft*{m} \< \\
    m^0 \from \{0,1\}^{\abs{m}} \<\< \\
    m^1 := m \<\< \\
    c \from \{0,1\} \<\< \\
    \sigma := \Enc(\sk; m^c) \<\< \\
    \< \sendmessageright*{\sigma} \< \\
    \<\< (A, B, \rho_{\calA,\calB}) \from \Adv(\sigma, \tau) \\
    \< \sendmessageleft*{A, B, \rho_{\calA,\calB}} \< \\
    (a', b') \from A^{\sk}\otimes B^{\sk}\,\rho_{\calA \calB} \<\< \\
    \textnormal{Output 1 if $a'=b'=c$; otherwise output 0} \<\<
}
\caption{$\UEExpt_{\enc,\Adv}(\secpar)$. The challenger samples a secret encryption key $\sk$, while the adversary decides on a message $m$ and sends it to the challenger. The resulting internal state of the adversary is $\tau$, which will be provided to the next part of the adversary. The challenger samples a fresh random message $m^0$, sets $m^1 := m$, and encrypts $m^c$ for $c \from \{0,1\}$ using $\sk$. The challenger sends the encryption $\sigma$ to the adversary, who maps this to a state $\rho_{\calA \calB}$ on the two registers $\calA, \calB$ and returns $\rho_{\calA \calB}$ to the challenger, together with descriptions of (families of) quantum circuits $A$ and $B$ on $\calA$ and $\calB$, respectively, indexed by keys. The challenger runs $A^\sk$ and $B^\sk$ on $\rho_{\calA \calB}$, obtaining outcomes $a'$ and $b'$. The adversary wins if $a'=b'=c$.}\label{game:UE}
\end{figure}

\begin{definition} \label{def:ue}
    A pair\footnote{Note that an alternative definition might involve a third algorithm for generating keys. However, we require that the keys are sampled uniformly at random as this will make our $\qsiO$ applications simpler.} of efficient quantum algorithms $(\enc,\dec)$ is an \emph{unclonable encryption (UE) scheme} if it satisfies the following conditions, for all $\lambda, n \in \mathbb{N}$:
    \begin{itemize}
        \item (Correctness) For all $\sk \in \{0,1\}^\lambda$ and $m \in \{0,1\}^n$,
        \begin{align*}
            \dec(\sk; \enc(\sk; m)) &\to m
        \end{align*}
        \item (Security) For all polynomial-time adversaries $\Adv$,
        \[
            \Pr[\emph{\UEExpt}_{\enc,\Adv}(\secpar) = 1] \le \frac{1}{2} + \negl \,,
        \]
        where $\emph{\UEExpt}_{\enc,\Adv}(\secpar)$ is defined in \Cref{game:UE}.
    \end{itemize}
\end{definition}

Let $\{\PRG_\secpar\}_{\secpar \in \mathbb{N}}$ be a family of pseudo-random generators with 1 bit of stretch. For $n > \secpar$, define $\PRG_{\secpar, n} : \{0,1\}^\secpar \to \{0,1\}^n$ to be $\PRG_{n-1} \circ \cdots \circ \PRG_{\secpar+1} \circ \PRG_{\secpar}$, the $(n-\secpar)$-fold composition of $\PRG$. For $n \le \secpar$, define $\PRG_{\secpar, n} : \{0,1\}^\secpar \to \{0,1\}^n$ to be the restriction to the first $n$ coordinates of the input. Then $\{\PRG_{\secpar,n}\}_{\secpar, n \in \mathbb{N}}$ is a family of pseudo-random generators of arbitrary stretch. The following is a natural candidate UE scheme:

\indent $\enc(\sk; m)$:
\begin{enumerate}
    \item Parse $\sk$ as $(\theta, U) \in \{0,1\}^{11 \secpar'} \times \{0,1\}^{\secpar' \times 11\secpar'}$, where $\secpar'$ is the largest integer such that $11 (\secpar')^2 + 11 \secpar' \le \secpar$.
    \item Sample a random string $x \from \{0,1\}^{11 \secpar'}$.
    \item Output $(\ket*{x^\theta}, m \oplus \PRG_{\secpar',\abs{m}}(Ux))$, where
    \begin{itemize}
        \item $Ux$ denotes the matrix-vector product over $\mathbb{F}_2$, and
        \item $\ket*{x^\theta}$ is $H^\theta \ket*{x}$, where $H^\theta$ denotes Hadamard gates applied to the qubits where the corresponding bit in $\theta$ is 1.
    \end{itemize}
\end{enumerate}

The decryption algorithm simply reads $x$ (since $\sk$ includes $\theta$) and computes
\[
    \Big(m \oplus \PRG_{\secpar',\abs{m}}(Ux)\Big) \oplus \PRG_{\secpar',\abs{m}}(Ux) = m.
\]
Note that we require the pseudo-random generator because, in our definition, the secret key is of a fixed length $\secpar$ whereas the length of the message is determined by the adversary. If we were satisfied with encrypting fixed-length messages (with a secret key that grows with the message length), then we could build cUE unconditionally.


\begin{figure}
\pcb{
    \textbf{Challenger} \< \< \textbf{Adversary} \\[][\hline]
    \< \< \\[-0.5\baselineskip]
    \sk_A, \sk_B \from \{0,1\}^\secpar \<\< (m_A, m_B, \tau) \from \Adv(\secparam) \\
    \< \sendmessageleft*{m_A, m_B} \< \\
    m_A^0 \from \{0,1\}^{\abs{m_A}},\ m_B^0 \from \{0,1\}^{\abs{m_B}} \<\< \\
    m_A^1 := m_A,\ m_B^1 := m_B \<\< \\
    a, b \from \{0,1\} \<\< \\
    \sigma := \Enc(\sk_A,\sk_B; m_A^a, m_B^b) \<\< \\
    \< \sendmessageright*{\sigma} \< \\
    \<\< (A, B, \rho_{\calA,\calB}) \from \Adv(\sigma, \tau) \\
    \< \sendmessageleft*{A, B, \rho_{\calA,\calB}} \< \\
    (a',b') \from A^{\sk_A} \otimes B^{\sk_B} \,\rho_{\calA,\calB} \<\< \\
    \textnormal{Output 1 if $a'=a$ and $b'=b$; otherwise output 0} \<\<
}
\caption{$\cUEExpt_{\enc,\Adv}(\secpar)$. The challenger samples encryption keys $\sk_A, \sk_B \from \{0,1\}^\secpar$. The adversary outputs messages $m_A, m_B$; its internal state $\tau$ will be used later. The challenger generates messages $m_A^0$, $m_B^0$, sets $m_A^1:=m_A$, $m_B^1:=m_B$, and randomly decides bits $a$, $b$. It encrypts $m_A^a, m_B^b$ with $\sk_A, \sk_B$ into $\sigma$ and sends it to the adversary. The adversary, with state $\sigma, \tau$, generates circuit descriptions $A$, $B$, and a state $\rho_{\calA,\calB}$, and sends them to the challenger. The challenger applies $A^{\sk_A}$ to $\rho_{\calA}$, giving $a'$, and $B^{\sk_B}$ to $\rho_{\calB}$, giving $b'$. The adversary wins if $a'=a$ and $b'=b$.} \label{game:cUE}
\end{figure}

\begin{definition} \label{def:cue}
    A pair of efficient quantum algorithms $(\enc,\dec)$ is a \emph{coupled unclonable encryption (cUE) scheme} if it satisfies the following conditions, for all $\lambda, n_A, n_B \in \mathbb{N}$:
    \begin{itemize}
        \item (Correctness) For all $\sk_A,\sk_B \in \{0,1\}^\lambda$ and $m_A \in \{0,1\}^{n_A}, m_B \in \{0,1\}^{n_B}$,
        \begin{align*}
            \dec(0, \sk_A; \enc(\sk_A, \sk_B; m_A, m_B)) &\to m_A \textnormal{ and} \\
            \dec(1, \sk_B; \enc(\sk_A, \sk_B; m_A, m_B)) &\to m_B.
        \end{align*}
        \item (Security) For all polynomial-time adversaries $\Adv$,
        \[
            \Pr[\emph{\cUEExpt}_{\enc,\Adv}(\secpar) = 1] \le \frac{1}{2} + \negl.
        \]
    \end{itemize}
\end{definition}

\begin{remark}
The reader may be wondering whether the security guarantees of UE or cUE imply standard CPA security. For UE, it is straightforward to see that \Cref{def:ue} implies CPA security: An adversary breaking CPA encryption can be used in the UE game to recover a guess for the challenge bit $c$. Then the UE adversary can simply set $A$ and $B$ to be families of circuits that always output $c$. On the other hand, for cUE (\Cref{def:cue}) the natural reduction implies that no adversary can simultaneously guess both challenges --- leaving open the possibility that the adversary can guess one of the challenges. It is therefore not clear whether cUE security implies CPA security for each message separately.
\end{remark}

\Cref{theorem:rand-indep} about unclonable randomness gets us most of the way towards building cUE. The natural approach to construct cUE is to use the unclonable randomness as a one-time pad for the adversary's chosen message. However, there are two small technical issues. First, in cUE the keys $\sk_A$ and $\sk_B$ must be sampled independently, but the keys $(\theta, U)$ and $(\theta, V)$ in the unclonable randomness game cannot be sampled independently because they both contain $\theta$. Second, the length of the message is determined by the adversary in the cUE game, whereas unclonable randomness has a fixed-length message as a function of $\secpar$. Therefore, our cUE scheme is slightly more complex than our unclonable randomness scheme, and additionally uses a pseudorandom generator.

We note that the matrix $T$ in \Cref{const:cUE} just serves to make the keys $\sk_A, \sk_B$ independent. If we were satisfied with the keys being partly identical (on $\theta$) and partly independent (on $U, V$), then we would not need $T$.

\begin{const}
\label{const:cUE}
Let $\{\PRG_{\secpar,n} : \{0,1\}^\secpar \to \{0,1\}^n\}_{\secpar,n \in \mathbb{N}}$ be a family of pseudorandom generators. Define $\enc$ and $\dec$ as follows.

\begin{itemize}
\item $\enc(\sk_A, \sk_B; m_A, m_B)$:
\indent 
    \begin{enumerate}
        \item Let $\secpar'$ be the largest integer such that $11 (\secpar')^2 + 11 \secpar' + 1 \le \secpar$. We parse the secret keys as $\sk_A = (\theta_A, U, -)$ and $\sk_B = (\theta_B, V, -)$, where $U, V \in \{0,1\}^{\secpar' \times (11\secpar')}$ and $\theta_A, \theta_B \in \{0,1\}^{11\secpar'+1}$. If there are any leftover bits in $\sk_A, \sk_B$ we discard them.
        \item Sample $x \from \{0,1\}^{11\secpar'}$ and $T \from \{0,1\}^{(11\secpar') \times (11\secpar'+1)}$ conditioned on $T \theta_A = T \theta_B$ and $\rank(T) = 11\secpar'$. Define $\theta := T \theta_A = T \theta_B$.
        \item Output $\ket*{x^\theta}, T, m_A \oplus \PRG_{\secpar',\abs{m_A}}(Ux), m_B \oplus \PRG_{\secpar',\abs{m_B}}(Vx)$.
    \end{enumerate}
\item $\dec(p, \sk; c)$: Parse $\sk$ as $\sk = (\theta', W, -)$ (again discarding any leftover bits). Parse $c$ as $c = (\ket{\psi}, T, c_0, c_1)$. Compute $\theta = T \theta'$. Apply $H^{\theta}$ to $\ket{\psi}$, and then measure in the standard basis to obtain an outcome $x$. Output $c_p \oplus \PRG_{\secpar',\abs{c_p}}(W x)$.
\end{itemize}   
\end{const}

\begin{theorem} \label{theorem:coupled-ue}
One-way functions imply the existence of cUE. In particular, Construction 2 is cUE.
\end{theorem}
\begin{proof}
We reduce security of Construction \ref{const:cUE} to security of unclonable randomness (recall that the latter is defined via $\RandExpt$ from Figure \ref{game:rand}).
    Let $\Adv$ be an efficient adversary for the security experiment $\cUEExpt$ for cUE, and let $\secpar$ be a security parameter. Let $\secpar'$ be the largest integer such that $11 (\secpar')^2 + 11 \secpar' + 1 \le \secpar$. We define an adversary $\Red[\Adv]$ for $\RandExpt(\secpar',\secpar')$ as follows.

    \indent $\RandExpt_{\Red[\Adv]}(\secpar', \secpar')$:
    \begin{enumerate}
        \item The challenger samples $a,b \from \{0,1\}$, $x, \theta \from \{0,1\}^{11\secpar'}$, $r^0, s^0 \from \{0,1\}^{\secpar'}$, and $U,V \from \{0,1\}^{\secpar' \times (11\secpar')}$. Let $r^1 := Ux$ and $s^1 := Vx$.
        \item The challenger sends $\ket*{x^\theta}, r^a, s^b$ to $\Red$.
        \item $\Red$ samples messages $(m_A, m_B) \from \Adv(\secparam)$ and a uniformly random rank-$(11\secpar')$ matrix $T$ from $\{0,1\}^{(11\secpar') \times (11\secpar'+1)}$.
        \item $\Red$ runs $(A, B, \rho_{\calA,\calB}) \from \Adv(\ket*{x^\theta}, T, m_A \oplus \PRG_{\secpar',\abs{m_A}}(r^a), m_B \oplus \PRG_{\secpar',\abs{m_B}}(s^b))$.
        \item For $d \in \{0,1\}, \theta' \in \{0,1\}^{11\secpar'}, U' \in \{0,1\}^{\secpar' \times 11\secpar'}$, define the circuit $\tilde{A}^{\theta',U'}_d$ as follows:
        \begin{enumerate}
            \item Let $\theta_A, \theta_B \in \{0,1\}^{11\secpar'+1}$ be the two vectors such that $T \theta_A = T \theta_B = \theta'$. Let $\theta_A$ be whichever vector has the first differing bit between $\theta_A$ and $\theta_B$ equal to $d$; similarly let $\theta_B$ be the vector which has $1-d$ at the first differing location.
            \item Let $\sk_A := (\theta_A, U', \texttt{pad})$ where $\texttt{pad}$ is a random string of length $\lambda - 11 (\secpar')^2 - 11 \secpar' - 1$.
            \item Return the output of running $A^{\sk_A}$ on the input state.
        \end{enumerate}
        Define $\tilde{B}^{\theta',V'}_d$ similarly.
        \item $\Red$ samples $d \from \{0,1\}$ and sends $\rho_{\calA,\calB}$ and $\tilde{A}_d$, $\tilde{B}_d$ to the challenger.
        \item The challenger measures $\tilde{A}_d^{\theta, U}$ and $\tilde{B}_d^{\theta, V}$ on $\rho_{\calA,\calB}$, obtaining outcomes $a'$ and $b'$. The reduction wins if $a'=a$ and $b'=b$.
    \end{enumerate}
    The view of $\Adv$ and $A,B$ in $\RandExpt_{\Red[\Adv]}(\secpar',\secpar')$ is computationally indistinguishable from that in $\cUEExpt_{\enc,\Adv}(\secpar)$ by security of the PRG. Therefore,
    \begin{align*}
        \Pr[\cUEExpt_{\enc,\Adv}(\secpar)=1] &\le \Pr[\RandExpt_{\Red[\Adv]}(\secpar', \secpar')=1] + \negl \\
        &\le \frac{1}{2} + \negl
    \end{align*}
    where we have invoked \Cref{theorem:rand-indep} for the second inequality.
\end{proof}

A direct inspection of the proof of \Cref{theorem:coupled-ue} gives the following.
\begin{corollary}
cUE exists unconditionally, for messages of fixed length.
\end{corollary}
\begin{proof}
The construction is identical to \Cref{const:cUE}, except that $Ux$ and $Vx$ are used directly as one-time pads, without first applying a PRG. Since the messages are of fixed length, one can sample $U$ and $V$ of the appropriate size. The security reduction is analogous to that for \Cref{theorem:coupled-ue}.
\end{proof}

\subsection{Key testing} \label{subsec:key-testing}
For our applications it will be important that our UE and cUE schemes have an additional property that we call \emph{key testing}. This states that there should exist an efficient algorithm $\test$ that determines whether a given secret key is ``correct'' for a given encryption.

\begin{definition}
    A triple of efficient quantum algorithms $(\enc,\dec,\test)$ is an \emph{unclonable encryption scheme with key testing} if it satisfies the following conditions, for all $\lambda, n \in \mathbb{N}$:
    \begin{itemize}
        \item (Correctness) For all $\sk \in \{0,1\}^\lambda$ and $m \in \{0,1\}^n$,
        \begin{align*}
            \dec(\sk; \enc(\sk; m)) &\to m.
        \end{align*}
        \item (Security) For all polynomial-time adversaries $\Adv$,
        \[
            \Pr[\emph{\UEExpt}_{\enc,\Adv}(\secpar)=1] \le \frac{1}{2} + \negl.
        \]
        \item (Key testing) For all $\sk,\sk' \in \{0,1\}^\lambda$ and $m \in \{0,1\}^n$,
        \begin{align*}
            \test(\sk'; \enc(\sk; m)) &\to \delta_{\sk}(\sk') \,,
        \end{align*}
        where $\delta_{\sk}$ is the indicator function that is $1$ only at $\sk$.
    \end{itemize}
\end{definition}

\begin{definition}
    A triple of efficient quantum algorithms $(\enc,\dec,\test)$ is a \emph{coupled unclonable encryption scheme with key testing} if it satisfies the following conditions, for all $\lambda, n_A, n_B \in \mathbb{N}$:
    \begin{itemize}
        \item (Correctness) For all $\sk_A, \sk_B \in \{0,1\}^\lambda$ and $m_A \in \{0,1\}^{n_A}, m_B \in \{0,1\}^{n_B}$,
        \begin{align*}
            \dec(0, \sk_A; \enc(\sk_A, \sk_B; m_A, m_B)) &\to m_A \textnormal{ and} \\
            \dec(1, \sk_B; \enc(\sk_A, \sk_B; m_A, m_B)) &\to m_B.
        \end{align*}
        \item (Security) For all polynomial-time adversaries $\Adv$,
        \[
            \Pr[\emph{\cUEExpt}_{\enc,\Adv}(\secpar)=1] \le \frac{1}{2} + \negl.
        \]
        \item (Key testing) For all $\sk_A, \sk_B,\sk' \in \{0,1\}^\lambda$ and $m_A \in \{0,1\}^{n_A}, m_B \in \{0,1\}^{n_B}$,
        \begin{align*}
            \test(\sk'; \enc(\sk_A, \sk_B; m_A, m_B)) &\to \begin{cases}
                0, & \sk' = \sk_A \\
                1, & \sk' = \sk_B \\
                \bot, & \sk' \not\in \{\sk_A,\sk_B\}
            \end{cases} \,.
        \end{align*}
    \end{itemize}
\end{definition}


We can add key testing to any (coupled) unclonable encryption scheme using $\qsiO$ and injective one-way functions. The same construction and proof also work with classical indistinguishability obfuscation in place of $\qsiO$, but we only state the result for $\qsiO$ because all of our applications use it.\footnote{Note that $\qsiO$ does \emph{not} trivially imply $\iO$ since $\qsiO$ outputs a quantum implementation.} The main idea to upgrade a UE or cUE scheme to one with key testing is to append to the ciphertext a $\qsiO$ obfuscation of the program $\delta_{\sk}$ (which is zero everywhere except at $\sk$). Intuitively, this allows one to test the validity of a secret key, while at the same time preserving unclonability thanks to the properties of $\qsiO$.

\begin{theorem} \label{theorem:key-testing}
    If injective one-way functions and $\qsiO$ exist, then any UE or cUE scheme can be compiled into one with key testing.
\end{theorem}
\begin{proof}
    For simplicity we only describe the compiler and proof for UE. The compiler for cUE is analogous.
    
    Let $(\enc,\dec)$ be a UE scheme. We build a UE scheme with key testing $(\enc', \dec', \test)$ as follows:
    \begin{align*}
        \enc'(s; m) &= (A, \enc(A s; m), \qsiO(\delta_s)) \\
        \dec'(s; (A, \sigma, \tau)) &= \dec(A s; \sigma) \\
        \test(s; (A, \sigma, \tau)) &= \Eval(\tau, s) \,,
    \end{align*}
    where $\Eval$ is a universal quantum evaluation circuit, $A \from \F_2^{\secpar \times 3 \secpar}$ is a random matrix sampled by $\enc'$, and the secret key $\sk = s$ is interpreted as a vector in $\F_2^{3\secpar}$. 
    
    Correctness and key testing are clear from the construction, so we turn to proving UE security.

    For a circuit $f : \F_2^{3\secpar} \to \{0,1\}$, let $P[f] : \F_2^{3\secpar} \to \{0,1\}$ be defined by $P[f](x) = (1-\delta_0(x)) \cdot f(x)$. That is, $P[f]$ is a circuit that outputs $f(x)$ for all $x$ except 0, on which $P[f]$ always outputs 0.
    
    By the $\qsiO$ guarantee, $\qsiO(\delta_s) \approx \qsiO(P[\delta_{\{0,s\}}])$. By Zhandry's subspace-hiding obfuscation result \cite{Zha21}, which assumes the existence of injective one-way functions, we have $\qsiO(P[\delta_{\{0,s\}}]) \approx \qsiO(P[\delta_T])$ for a random subspace $T \subseteq \F_2^{3 \secpar}$ of dimension $2\secpar$ that contains $s$.

    Conditioned on $T$, observe that $s$ still has $2\secpar$ bits of min-entropy. Therefore the leftover hash lemma implies that $(A, T, A s)$ is $\negl$-close in statistical distance to $(A, T, u)$ for $u \from \F_2^\secpar$, so
    \[
        (A, \enc(A s; m), \qsiO(P[\delta_T])) \equiv (A, \enc(u; m), \qsiO(P[\delta_T])).
    \]
    Since $A$ and $\qsiO(P[\delta_T])$ can be sampled independently from $\enc(u; m)$, the UE security of $(\enc',\dec',\test)$ follows from the UE security of $(\enc,\dec)$.
\end{proof}

\section{Copy Protection}
\label{sec:cp}
In \Cref{sec:qsiO}, we showed that $\qsiO$ is ``'best-possible'' copy protection, and thus provides a principled heuristic for copy-protecting any functionality. In this section, our goal is to investigate which functionalities are \emph{provably} copy protected by $\qsiO$. We consider copy protection for three classes of functions, each with slightly different copy protection guarantees. All three security games begin with the challenger sending the adversary a quantum state that represents some copy-protected functionality; the adversary then applies some quantum channel to the received state, and creates a new state on two registers. The three security games differ from this point on:
\begin{enumerate}
    \item In \emph{decision} copy protection, each part of the adversary is given a uniformly random challenge input $x$, along with either (a) $f(x)$, or (b) $f(x')$ for a fresh random $x'$. The task is for both parts to correctly guess which case they are in.
    \item In \emph{search} copy protection, each part of the adversary is given a uniformly random input $x$, and asked to produce $y$ satisfying some condition $\Ver(x,y)$.
    \item In copy protection for point functions, each part of the adversary is given both the marked input and a uniformly random input. The task is for both parts to correctly guess which one is the marked input.
\end{enumerate}

Whereas point functions are a particular class of functions, the notions of decision and search copy protection are applicable to many classes of functions. We show that the classes of ``decision puncturable'' and ``search puncturable'' programs can be decision copy protected and search copy protected, respectively. Roughly, a \emph{decision puncturable program} does not reveal any information about the function value at the punctured point; a \emph{search puncturable program} may reveal some information, but an efficient adversary cannot compute from it any output that passes some (public or private) verification procedure at the punctured point. We define these notions of puncturable programs precisely in Section \ref{sec:puncturable-programs}.

Informally, our main results of this section are:
\begin{enumerate}
    \item Assuming injective OWFs, $\qsiO$ decision-copy-protects any decision-puncturable program.
    \item Assuming injective OWFs and UE, $\qsiO$ search-copy-protects any search-puncturable program.
    \item Assuming injective OWFs and UE, $\qsiO$ copy-protects point functions.
\end{enumerate}

\begin{remark}
    For clarity of presentation we assume throughout this section that all challenge input distributions in the copy protection security games are uniform. These results can be generalized to arbitrary distributions with high min-entropy using a randomness extractor.
\end{remark}

\subsection{Puncturable programs}
\label{sec:puncturable-programs}
A puncturing procedure for a class of programs $\calF$ is an efficient algorithm $\Puncture$ that takes as input a description of a program $f \in \calF$ and polynomially-many points $x_1, \dots, x_t \in \textnormal{Domain}(f)$, and outputs the description of a new program $f_{x_1, \dots, x_t}$. This program should satisfy $f_{x_1, \dots, x_t}(z) = f(z)$ for all $z \in \textnormal{Domain}(f) \setminus \{x_1, \dots, x_t\}$ as well as an additional security property:
\begin{itemize}
    \item For \emph{decision puncturing}, we require $(f_x, f(x)) \approx (f_x, f(x'))$ for a random $x'$. For instance, in \cite{SW21} it was shown that one-way functions imply the existence of decision puncturable pseudo-random functions.
    \item For \emph{search puncturing}, we require that no efficient adversary can compute, given $f_x$, an output $y$ such that $\Ver(f,x,y) = 1$, for some efficient (public or private) verification procedure $\Ver$. For example, if $f$ is a signing function with a hard-coded secret key or a message authentication code, $\Ver(f,x,y)$ would use the verification key to check that $y$ is a valid signature or authentication tag for $x$. In \cite{BSW16}, puncturable signatures were constructed from injective one-way functions and (classical) indistinguishability obfuscation.
\end{itemize}

\begin{definition}[Decision puncturable programs, \cite{SW21}] \label{def:decision-puncturing}
    Let $\calF = \{\calF_\secpar : \{0,1\}^{n(\secpar)} \to \{0,1\}^{m(\secpar)}\}_{\secpar \in \mathbb{N}}$ be a family of classical circuits. We say that $\calF$ is \emph{decision puncturable} if there exists an efficient algorithm $\Puncture$ such that, for each $\secpar \in \mathbb{N}$,
    \begin{itemize}
        \item For every $f \in \calF_\secpar$ and all $\poly(\secpar)$-sized sets $S \subseteq \{0,1\}^{n(\secpar)}$, $\Puncture(f,S)$ outputs a $\poly(\secpar)$-sized circuit $f_S$ such that for all $x \in \{0,1\}^{n(\secpar)} \setminus S$, $f_S(x) = f(x)$.
        \item For every QPT adversary $(\Adv_1, \Adv_2)$ such that $\Adv_1(\secparam)$ outputs a set $S \subseteq \{0,1\}^{n(\secpar)}$ and a state $\sigma$, if $f \from \calF_\secpar$, $f_S \from \Puncture(f,S)$, and $\hat{S} \subseteq \{0,1\}^{n(\secpar)}$ is a uniformly random set of the same size as $S$,
        \[
            \abs{\Pr[\Adv_2(\sigma, f_S, S, f(S)) = 1] - \Pr[\Adv_2(\sigma, f_S, S, f(\hat{S}) = 1]} = \negl \,,
        \]
        where $f(S) = \{f(x): x \in S\}$, and similarly for $f(\hat{S})$.
    \end{itemize}
\end{definition}

We only require search puncturable programs to be puncturable at a single point, because this definition suffices for our applications. This is also the definition given in \cite{BSW16}.

\begin{definition}[Search puncturable programs] \label{def:search-puncturing}
    Let $\calF = \{\calF_\secpar : \{0,1\}^{n(\secpar)} \to \{0,1\}^{m(\secpar)}\}_{\secpar \in \mathbb{N}}$ and $\Ver = \{\Ver_\secpar : \calF_\secpar \times \{0,1\}^{n(\secpar)} \times \{0,1\}^{m(\secpar)}\}_{\secpar \in \mathbb{N}}$ be families of $\poly(\secpar)$-sized classical circuits. We say that $\calF$ is \emph{search puncturable} with respect to $\Ver$ if, for each $\secpar \in \mathbb{N}$, there exists an efficient algorithm $\Puncture$ such that
    \begin{itemize}
        \item For every $f \in \calF_\secpar$ and all $x \in \{0,1\}^{n(\secpar)}$, $\Puncture(f,x)$ outputs a $\poly(\secpar)$-sized circuit $f_x$ such that for all $x' \in \{0,1\}^{n(\secpar)} \setminus \{x\}$, $f_S(x') = f(x')$.
        \item For every QPT adversary $(\Adv_1, \Adv_2)$ such that $\Adv_1(\secparam)$ outputs a point $x \in \{0,1\}^{n(\secpar)}$ and a state $\sigma$, if $f \from \calF_\secpar$ and $f_x \from \Puncture(f,x)$,
        \[
            \Pr[\Ver_\secpar(f, x, \Adv_2(\sigma, f_x, x)) = 1] = \negl \,.
        \]
    \end{itemize}
\end{definition}

\subsection{Decision copy protection} \label{subsec:decision-cp}
All of the copy protection variants that we define in Section \ref{sec:cp} have the same correctness definition (Definition \ref{def:cp-correctness}). They only differ in their definition of security. 



\begin{figure}[H]
\pcb{
    \textbf{Challenger} \< \< \textbf{Adversary} \\[][\hline]
    \< \< \\[-0.5\baselineskip]
    f \from \calF_\secpar \<\< \\
    x_A, x_B, x_A', x_B' \from \{0,1\}^{\secpar} \<\< \\
    y_A^0 := f(x_A'),\ y_B^0 := f(x_B') \<\< \\
    y_A^1 := f(x_A),\ y_B^1 := f(x_B) \<\< \\
    \sigma := \CP(f) \<\< \\
    \< \sendmessageright*{\sigma} \< \\
    \<\< (A, B, \rho_{\calA,\calB}) \from \Adv(\sigma) \\
    \< \sendmessageleft*{A, B, \rho_{\calA,\calB}} \< \\
    a,b \from \{0,1\} \<\< \\
    (a',b') \from A^{x_A,y_A^a} \otimes B^{x_B,y_B^b} \,\rho_{\calA\calB} \<\< \\
    \textnormal{Output 1 if $a'=a$ and $b'=b$; otherwise output 0}
}
\caption{$\CPExptDecision_{\CP,\Adv,\calF}(\secpar)$. The challenger samples a function $f$ from $\calF_\secpar$ and $x_A, x_B, x_A', x_B' \from \{0,1\}^{\secpar}$. It computes function values for pairs $(x_A, x_A')$, $(x_B, x_B')$ to produce challenge pairs, and sends the copy protection $\CP(f)$ to the adversary. The adversary returns a quantum state and quantum circuit descriptions $A$ and $B$. The challenger measures these with chosen challenge pairs, deciding which function value to use based on random bits $a$, $b$. The adversary wins if the measured values match $a$, $b$.} \label{game:CP-decision}
\end{figure}

\begin{definition}[Decision copy protection security] \label{definition:CP-decision}
Let $\calF = \{\calF_{\secpar} : \{0,1\}^\secpar \to \{0,1\}^{m(\secpar)}\}_{\secpar \in \mathbb{N}}$ be a family of $\poly(\secpar)$-sized classical circuits. Let $\CP$ be a copy protection scheme for $\mathcal{F}$ (as in Definition \ref{def:cp-correctness}). We say that $\CP$ is decision copy protection secure if, for all QPT algorithms $\Adv$, there exists a negligible function $\mathsf{\negl}$ such that, for all $\secpar$,
    \[
        \Pr[\emph{\CPExptDecision}_{\CP,\Adv,\calF}(\secpar)=1] \le \frac{1}{2} + \negl \,,
    \]
where \textnormal{$\CPExptDecision_{\CP,\Adv,\calF}(\secpar)$} is defined in \Cref{game:CP-decision}.
\end{definition}

\begin{theorem} \label{theorem:decision-cp}
    Let $\qsiO$ be a secure $\qsiO$ scheme. Let $\calF$ be any family of decision-puncturable programs. Then, assuming injective one-way functions exist, $\qsiO$ is a copy protection scheme for $\calF$ that is decision copy protection secure.
\end{theorem}
\begin{proof}
    The proof proceeds via a sequence of hybrids. Fix an adversary $\Adv$ for $\CPExptDecision_{\qsiO,\Adv,\calF}(\secpar)$. We will show that the success probability of $\Adv$ is preserved across the hybrids, up to $\negl$. We will then argue that the final hybrid is secure by invoking the security of a cUE scheme built from an injective one-way function.

    \vspace{2mm}
    \noindent \textit{Hybrid 0}: The original security game $\CPExptDecision_{\qsiO,\Adv,\calF}(\secpar)$, for random $a,b$.
    
    \vspace{2mm}
    \noindent \textit{Hybrid 1}: Let $(\enc,\dec,\test)$ be any cUE scheme with key-testing. By \Cref{theorem:key-testing}, such a scheme can be built from $\qsiO$ and any injective one-way function. Let $P[g,\sigma]$ be the following program, for a classical circuit $g$ and a quantum state $\sigma$ (which is supposed to be an output of $\enc$):
    
    \indent $P[g,\sigma](z)$:
    \begin{enumerate}
        \item Compute $\test(z; \sigma) \to r$. If $r=\bot$, terminate and output $g(z)$; otherwise continue to step 2.
        \item Compute $\Dec(r, z; \sigma)$ and output the result.
    \end{enumerate}
    Hybrid 1 is same as Hybrid 0, except the challenger sends
    \[
        \qsiO(P[f_{x_A,x_B}, \Enc(x_A,x_B;f(x_A),f(x_B))])
    \]
    as the first message instead of $\qsiO(f)$, where $f_{x_A,x_B} \leftarrow \Puncture(f,\{x_A, x_B\})$. Crucially, notice that $P[f_{x_A,x_B}, \Enc(x_A,x_B;f(x_A),f(x_B))]$ is functionally equivalent to $f$. More precisely, it can be viewed as quantum implementation of $f$ (as in Definition \ref{def:quantum-implementation}). Then, the fact that $\Adv$ wins Hybrid 1 with probability at most $\negl$ higher than in Hybrid 0 follows directly from the security guarantee of $\qsiO$.

    Note that if there was a bound on the size of the messages that $\enc$ could encrypt in terms of the length of the secret keys, then we would only be able to encrypt the output of functions $f : \{0,1\}^\secpar \to \{0,1\}^{m(\secpar)}$ with sufficiently small output length $m(\secpar)$. Fortunately, in the definition of cUE we allow the adversary to choose the length of the message (in this case $m(\secpar)$).

    \vspace{2mm}
    \noindent \textit{Hybrid 2}: The challenger samples $\tilde{x}_A, \tilde{x}_B \from \{0,1\}^{\secpar}$, sends
    \[
        \qsiO(P[f_{x_A,x_B},\Enc(x_A,x_B;f(\tilde{x}_A),f(\tilde{x}_B))])
    \]
    as the first message, and uses $y_A^1 := f(\tilde{x}_A), y_B^1 := f(\tilde{x}_B)$ instead of $y_A^1 := f(x_A), y_B^1 := f(x_B)$. Suppose for a contradiction that an adversary $\Adv$ had non-negligibly higher success probability in Hybrid 2 than in Hybrid 1. Then, we can use $\Adv$ to break the decision puncturing security (\Cref{def:decision-puncturing}) of $f$ as follows: 
    \begin{itemize}
        \item Sample $x_A,x_B \leftarrow \{0,1\}^{\lambda}$.
        \item Receive $(f_{x_A,x_B}, y_A^1, y_B^1)$ where either
        $(y_A^1, y_B^1) = (f(x_A), f(x_B))$ or $(y_A^1, y_B^1) = (f(\tilde{x}_A), f(\tilde{x}_B))$ for uniformly random $\tilde{x}_A, \tilde{x}_B$.
        \item Simulate the rest of the Hybrid 1 experiment with $\Adv$ using $(f_{x_A,x_B}, y_A^1, y_B^1)$.
    \end{itemize}
    When $(y_A^1, y_B^1) = (f(x_A), f(x_B))$, the simulated experiment is distributed as in Hybrid 1. When $(y_A^1, y_B^1) = (f(\tilde{x}_A), f(\tilde{x}_B))$, it is distributed as in Hybrid 2. 

    \vspace{2mm}
    \noindent \textit{Hybrid 3}: Let $\tilde{P}[g,\sigma]$ be the following program, for a classical circuit $g$ and a quantum state $\sigma$ (which is supposed to be an output of $\enc$).
    
    \indent $\tilde{P}[g,\sigma](z)$:
    \begin{enumerate}
        \item Compute $\test(z; \sigma) \to r$. If $r=\bot$, terminate and output $g(z)$; otherwise continue to step 2.
        \item Compute $\Dec(r, z; \sigma) \to y$ and output $g(y)$.
    \end{enumerate}
    Hybrid 3 is the same as Hybrid 2, except the challenger sends
    \[
        \qsiO(\tilde{P}[f, \Enc(x_A,x_B;\tilde{x}_A,\tilde{x}_B)])
    \]
    as the first message. That Hybrid 3 has the same success probability as Hybrid 2 follows from the security guarantee of $\qsiO$.

    We complete the proof by giving a reduction from an adversary $\Adv$ for Hybrid 3 to an adversary $\Red$ for the cUE game.
    
    \indent $\cUEExpt_{\enc,\Red[\Adv]}(\secpar)$:
    \begin{enumerate}
        \item $\Red$ samples $\tilde{x}_A, \tilde{x}_B \from \{0,1\}^\secpar$ and sends these to the challenger.
        \item The challenger samples $x_A, x_B \from \{0,1\}^\secpar$ and $m_A^0, m_B^0 \from \{0,1\}^\secpar$, and lets $m_A^1 = \tilde{x}_A, m_B^1 = \tilde{x}_B$.
        \item The challenger samples $a,b \from \{0,1\}$, computes $\sigma = \enc(x_A, x_B; m_A^a, m_B^b)$ and sends $\sigma$ to $\Red$.
        \item $\Red$ samples $f \from \calF_\secpar$ and sends $\qsiO(\tilde{P}[f, \sigma])$ to $\Adv$. Thus, $\Red$ obtains $(A,B,\rho_{\calA,\calB}) \from \Adv(\qsiO(\tilde{P}[f, \sigma]))$.
        \item Let $\tilde{A}^x = A^{x,f(\tilde{x}_A)}$ and $\tilde{B}^x = B^{x, f(\tilde{x}_B)}$.
        \item $\Red$ sends $(\tilde{A}, \tilde{B}, \rho_{\calA,\calB})$ to the challenger.
        \item The challenger measures $\tilde{A}$ and $\tilde{B}$ on $\rho_{\calA,\calB}$, obtaining $a'$ and $b'$. The reduction wins if $a'=a$ and $b'=b$.
    \end{enumerate}
    The view of $\Adv$ in $\cUEExpt_{\enc,\Red[\Adv]}(\secpar)$ is exactly the same as in Hybrid 3. Hence, the cUE security of $\enc$ implies that the success probability of $\Adv$ Hybrid 3 is at most $1/2+\negl$.
\end{proof}

\subsection{Search copy protection} \label{subsec:search-cp}

\begin{figure}[H]
\pcb{
    \textbf{Challenger} \<\< \textbf{Adversary} \\[][\hline]
    \< \< \\[-0.5\baselineskip]
    f \from \calF_\secpar \<\< \\
    \sigma := \CP(f) \<\< \\
    \< \sendmessageright*{\sigma} \< \\
    \<\< (A, B, \rho_{\calA,\calB}) \from \Adv(\sigma) \\
    \< \sendmessageleft*{A, B, \rho_{\calA,\calB}} \< \\
    x \from \{0,1\}^\secpar \<\< \\
    (y_A, y_B) \from A^{x} \otimes B^{x} \rho_{\calA\calB} \<\< \\
    \textnormal{Output 1 if $\Ver_\secpar(f, x, y_A) = 1$ and} \<\< \\
    \textnormal{\hspace{0.5cm} $\Ver_\secpar(f, x, y_B) = 1$; otherwise output 0} \<\<
}
\caption{$\CPExptSearch_{\CP,\Adv,\calF,\Ver}(\secpar)$. The challenger samples a function $f \from \calF_\secpar$ and sends $\CP(f)$ to the adversary. The adversary responds with two quantum circuits $A$, $B$ and a quantum state $\rho_{\calA,\calB}$. The challenger samples a random string $x$ and measures $A,B$ on $\rho_{\calA,\calB}$ to obtain $y_A$, $y_B$. The adversary wins if both $y_A$ and $y_B$ pass the verification using function $f$, input $x$, and the measured output.} \label{game:CP-search}
\end{figure}

\begin{definition} \label{definition:CP-search}
Let $\calF = \{\calF_\secpar : \{0,1\}^{\secpar} \to \{0,1\}^{m(\secpar)}\}_{\secpar \in \mathbb{N}}$ and $\Ver = \{\Ver_\secpar : \calF_\secpar \times \{0,1\}^{\secpar} \times \{0,1\}^{m(\secpar)}\}_{\secpar \in \mathbb{N}}$ be families of $\poly(\secpar)$-sized classical circuits. Let $\CP$ be a copy protection scheme for $\mathcal{F}$ (as in Definition \ref{def:cp-correctness}). We say that $\CP$ is search copy protection secure with respect to $\Ver$ if, for all QPT algorithms $\Adv$, there exists a negligible function $\mathsf{\negl}$ such that, for all $\secpar$,
    \[
        \Pr[\emph{\CPExptSearch}_{\CP,\Adv,\calF,\Ver}(\secpar)=1] \le \negl \,,
    \]
where \textnormal{$\CPExptSearch_{\CP,\Adv,\calF,\Ver}(\secpar)$} is defined in \Cref{game:CP-search}.
\end{definition}

\begin{theorem} \label{theorem:search-cp}
Let $\qsiO$ be a secure $\qsiO$ scheme. Let $\calF$ be any family of search-puncturable programs with respect to $\Ver$. Then, assuming injective one-way functions and unclonable encryption exist, $\qsiO$ is a copy protection scheme for $\calF$ with respect to $\Ver$ that is search copy protection secure.
\end{theorem}
\begin{proof}
    Let $(\enc,\dec,\test)$ be any UE scheme with key-testing. By \Cref{theorem:key-testing}, such a scheme can be built from $\qsiO$, UE, and any injective one-way function. Let $P[g,\sigma]$ be the following program, for a classical circuit $g$ and a quantum state $\sigma$ (which is supposed to be an output of $\enc$):
    
    \indent $P[g,\sigma](z)$:
    \begin{enumerate}
        \item Compute $\test(z; \sigma)$. If it rejects, terminate and output $g(z)$; otherwise continue to step 2.
        \item Compute $\Dec(z; \sigma)$. If the first $\secpar$ bits of the decryption are not $0^\secpar$, terminate and output $\bot$; otherwise interpret the remaining bits as the description of a circuit $g'$ and output $g'(z)$.
    \end{enumerate}

    We reduce from the UE game to $\CPExptSearch$ as follows:
    
    \indent $\UEExpt_{\enc,\Red[\Adv]}(\secpar)$:
    \begin{enumerate}
        \item $\Red$ samples $f \from \calF_\secpar$ and sends $0^\secpar || f$ to the challenger. Let $\abs{f}$ be the number of bits in the description of $f$.
        \item The challenger samples $x \from \{0,1\}^\secpar$, $m^0 \from \{0,1\}^{\secpar + \abs{f}}$, and lets $m^1 = 0^\secpar || f$.
        \item The challenger samples $c \from \{0,1\}$, computes $\sigma = \enc(x; m^c)$ and sends $\sigma$ to $\Red$.
        \item $\Red$ sends $\qsiO(P[f, \sigma])$ to $\Adv$.
        \item $(A,B,\rho_{\calA,\calB}) \from \Adv(\qsiO(\tilde{P}[f, \sigma]))$
        \item $\Red$ samples a random bit $r$ and defines $\tilde{A}, \tilde{B}$ as follows:
        \begin{itemize}
            \item[$\tilde{A}^x$:] Run $A^x \to y$. If $\Ver(f,x,y) = 1$, output $a'=1$; otherwise output $a'=r$.
            \item[$\tilde{B}^x$:] Run $B^x \to y$. If $\Ver(f,x,y) = 1$, output $b'=1$; otherwise output $b'=r$.
        \end{itemize}
        \item $\Red$ sends $(\tilde{A}, \tilde{B}, \rho_{\calA,\calB})$ to the challenger.
        \item The challenger measures $\tilde{A}$ and $\tilde{B}$ on $\rho_{\calA,\calB}$, obtaining $a'$ and $b'$. The reduction wins if $a'=b'=c$.
    \end{enumerate}
    By security of $\Enc$, we have that
    \begin{equation} \label{ineq:scp-1}
        \Pr[\UEExpt_{\enc,\Red[\Adv]}(\secpar)=1] \le \frac{1}{2}+\negl.
    \end{equation}
    
    Now suppose that $\Adv$ wins $\CPExptSearch$ with probability $\varepsilon$. We consider the $c=1$ and $c=0$ cases separately. For $\tilde{c} \in \{0,1\}$, let $\UEExpt_{\enc,\Adv,\tilde{c}}$ denote the $c=\tilde{c}$ version of $\UEExpt_{\enc,\Adv}$.
    
    \paragraph{The $c=1$ case.} By the $\qsiO$ guarantee,
    \[
        \qsiO(f) \approx \qsiO(P[f, \enc(x; 0^\secpar || f)]),
    \]
    so both $A$ and $B$ output $y$'s such that $\Ver(f,x,y)=1$ with probability at least $\varepsilon - \negl$ in $\UEExpt_{\enc,\Red[\Adv]}(\secpar)$. By construction of $\tilde{A}, \tilde{B}$, it follows that
    \begin{equation} \label{ineq:scp-2}
        \Pr[\UEExpt_{\enc,\Red[\Adv],1}(\secpar)=1] \ge \varepsilon + \frac{1-\varepsilon}{2} - \negl = \frac{1+\varepsilon}{2} - \negl.
    \end{equation}
    
    \paragraph{The $c=0$ case.} When $c=0$, $\Red$ receives $\qsiO(P[f, \Enc(x; m^0)])$ for $m^0 \from \{0,1\}^{\secpar+\abs{f}}$. Since $m^0$ begins with something other than $0^\secpar$ with probability $1-\negl$, the $\qsiO$ guarantee implies that
    \[
        \qsiO(P[f, \Enc(x; m^0)]) \approx \qsiO(P[f_x, \Enc(x; m^0)]).
    \]
    By the search puncturing security of $f_x$ (\Cref{def:search-puncturing}), $\Adv(\qsiO(P[f_x, \Enc(x; m^0)]))$ cannot produce a $y$ satisfying $\Ver(f,x,y) = 1$ with probability greater than $\negl$. Therefore, $\tilde{A}^x, \tilde{B}^x$ both output $r$ with probability $1-\negl$, and
    \begin{equation} \label{ineq:scp-3}
        \Pr[\UEExpt_{\enc,\Red[\Adv],0}(\secpar)=1] \ge \frac{1}{2} - \negl.
    \end{equation}
    Putting together Inequalities \ref{ineq:scp-1}, \ref{ineq:scp-2}, and \ref{ineq:scp-3}, we find that
    \[
        \frac{1}{2} + \frac{\varepsilon}{4} \le \frac{1}{2} + \negl
    \]
    and therefore $\varepsilon \le \negl$.
\end{proof}

\subsection{Copy protection for point functions}
We show that $\qsiO$ and unclonable encryption with key-testing yield copy protection for point functions with perfect correctness. The precise security notion that is achieved here is possibly stronger than those considered in previous works, e.g. \cite{AMP20,AKL+22}. In particular, in the security game that we consider, each of Alice and Bob receive the \emph{same} challenge, which consists of \emph{both} the marked input and a uniformly random input (in a random order), and they have to guess which one is the marked input.

\begin{figure}[H]
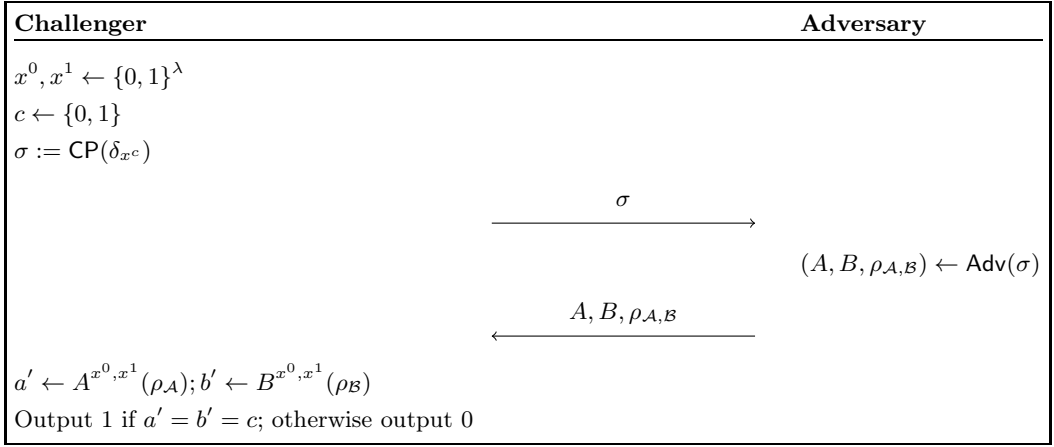

\pcb{
    \textbf{Challenger} \< \< \textbf{Adversary} \\[][\hline]
    \< \< \\[-0.5\baselineskip]
    x^0, x^1 \from \{0,1\}^\secpar \<\< \\
    c \from \{0,1\} \<\< \\
    \sigma := \CP(\delta_{x^c}) \<\< \\
    \< \sendmessageright*{\sigma} \< \\
    \<\< (A, B, \rho_{\calA,\calB}) \from \Adv(\sigma) \\
    \< \sendmessageleft*{A, B, \rho_{\calA,\calB}} \< \\
    a' \from A^{x^0,x^1}(\rho_{\calA}); b' \from B^{x^0,x^1}(\rho_{\calB}) \<\< \\
    \textnormal{Output 1 if $a'=b'=c$; otherwise output 0} \<\<
}
\caption{$\CPExptPF_{\CP,\Adv}(\secpar)$. The challenger samples two random strings $x^0, x^1 \from \{0,1\}^\secpar$, selects a random bit $c$, and sends $\CP(\delta_{x^c})$ to the adversary. The adversary generates a state $\rho_{\calA,\calB}$ and quantum circuit descriptions $A,B$ and sends them back. The challenger applies $A^{x^0,x^1}$ to $\rho_{\calA}$, giving $a'$, and $B^{x^0,x^1}$ to $\rho_{\calB}$, giving $b'$. The adversary wins if $a'=b'=c$.} \label{game:CP-PF}
\end{figure}

\begin{definition} \label{definition:CP-PF}
    An efficient algorithm $\CP$ is a copy protection scheme for point functions if, for all efficient adversaries $\Adv$,
    \[
        \Pr[\emph{\CPExptPF}_{\CP,\Adv}(\secpar)=1] \le \frac{1}{2} + \negl.
    \]
\end{definition}

\begin{theorem} \label{theorem:point-functions-cp}
    Assuming injective one-way functions and unclonable encryption, $\qsiO$ copy-protects point functions.
\end{theorem}

A new difficulty arises here that was not present when copy-protecting puncturable programs: For point functions, the copy protection security game is about distinguishing between \emph{inputs}. A naive application of the techniques from \Cref{theorem:decision-cp,theorem:search-cp} would therefore involve using UE to encrypt the secret keys. Instead, we will use UE to encrypt \emph{the first bit} at which the two challenges differ.

\begin{proof}
    Let $T$ be a rank-$\secpar$ matrix in $\F_2^{\secpar \times (\secpar+1)}$ and $\sigma$ be a quantum state. For input $z \in \F_2^{\secpar+1}$, we define a program $P_{T,\sigma}$ as follows.

    \indent $P_{T,\sigma}(z)$:
    \begin{enumerate}
        \item Compute $\test(Tz; \sigma)$. If it rejects, terminate and output 0.
        \item Compute $x^0 \ne x^1$ such that $Tx^0 = Tx^1 = Tz$. For $c \in \{0,1\}$, let
        \begin{equation} \label{eq:bit-diff}
            x(c) = \begin{cases}
                x^0, & (x^0)_i = c \textnormal{ where } i = \min\{j \in [\secpar+1] \mid (x^0)_j \ne (x^1)_j\} \\
                x^1, & \textnormal{ otherwise}.
            \end{cases}
        \end{equation}
        \item Compute $\Dec(Tz; \sigma) \to c$. If $x(c) = z$, output 1; otherwise output 0.
    \end{enumerate}

    Observe that $P_{T, \Enc(\sk; c)}(z) = \delta_{x(c)}(z)$, so $\qsiO(P_{T, \Enc(\sk; c)}) \approx \qsiO(\delta_{x(c)})$.

    We now describe a reduction $\Red$ that plays $\UEExpt$ using an adversary for the point function copy-protection game. We use a slight variant of $\UEExpt$ where the challenger encrypts a random bit, which is equivalent to the game presented in \Cref{game:UE} in the case of single bit messages.
    
    \indent $\UEExpt_{\enc,\Red[\Adv]}(\secpar)$:
    \begin{enumerate}
        \item The challenger samples $\sk \from \{0,1\}^\secpar$.
        \item The challenger samples $c \from \{0,1\}$, computes $\sigma = \enc(\sk; c)$ and sends $\sigma$ to $\Red$.
        \item $\Red$ samples a random rank-$\secpar$ matrix $T \from \F_2^{\secpar \times (\secpar+1)}$ and sends $\qsiO(P_{T, \sigma})$ to $\Adv$.
        \item $(A,B,\rho_{\calA,\calB}) \from \Adv(\qsiO(P_{T, \sigma}))$
        \item Let $\tilde{A}^\sk$ do the following on input $\rho_{\calA}$:
        \begin{enumerate}
            \item Compute $x^0 \ne x^1$ such that $Tx^0 = Tx^1 = \sk$.
            \item Run $A^{x^0,x^1}(\rho_{\calA}) \to a'$.
            \item Let $x(\cdot)$ be defined as in \eqref{eq:bit-diff}. Output $\tilde{a} \in \{0,1\}$ such that $x(\tilde{a}) = x^{a'}$.
        \end{enumerate}
        Define $\tilde{B}^\sk$ similarly but with $B, \rho_{\calB}$ instead.
        \item $\Red$ sends $(\tilde{A}, \tilde{B}, \rho_{\calA,\calB})$ to the challenger.
        \item The challenger measures $\tilde{A}$ and $\tilde{B}$ on $\rho_{\calA,\calB}$, obtaining $a'$ and $b'$. The reduction wins if $a'=b'=c$.
    \end{enumerate}
    By the security of $\qsiO$,
    \begin{align*}
        & \Pr[\CPExptPF_{\CP,\Adv}(\secpar+1)=1] \\
        &\le \Pr[\UEExpt_{\enc,\Red[\Adv]}(\secpar)=1] + \negl \\
    \intertext{and by the security of $\Enc$,}
        & \Pr[\UEExpt_{\enc,\Red[\Adv]}(\secpar)=1] \\
        &\le \frac{1}{2} + \negl. \qedhere
    \end{align*}
\end{proof}

\nocite{*}
\bibliographystyle{alpha}
\bibliography{references}

\end{document}